\documentclass[A4,12pt]{article}
\usepackage{amsmath,amssymb,amsfonts,amsthm,graphicx}
\usepackage{tkz-graph}
\usepackage{color}
\usepackage{cite}
\usepackage{epsfig}
\usepackage{epstopdf}
\usepackage{footnote}
\usepackage{cite}
\usepackage{caption}
\usepackage{subcaption}
\usepackage{enumerate}
\usepackage{algorithm}
\usepackage{algpseudocode}
\usepackage{times}
\usepackage[top=2cm, bottom=2cm, left=2cm, right=2cm]{geometry}

\newtheorem{theorem}{Theorem}[section]
\newtheorem{lemma}[theorem]{Lemma}

\theoremstyle{definition}
\newtheorem{definition}[theorem]{Definition}

\theoremstyle{remark}
\newtheorem{remark}[theorem]{\bf Remark}
\title{A Linear Algorithm for Computing  $\gamma_{[1,2]}$-set in Generalized Series-Parallel Graphs}
\author{P. Sharifani$^{1}$, M.R.~Hooshmandasl$^{2}$\\
	\footnotesize{$^{1,2}$Department of Computer Science, Yazd University, Yazd, Iran.}\\
	\footnotesize{$^{1,2}$The Laboratory of Quantum Information Processing, Yazd University, Yazd, Iran.}  \\
	\footnotesize{e-mail:$^1$ pouyeh.sharifani@gmail.com, $^2$ hooshmandasl@yazd.ac.ir}}
\date{}

\begin{document}

	\maketitle
	
	\begin{abstract}
		For a graph $G=(V,E)$, a set $S \subseteq V$ is a $[1,2]$-set if it is a dominating set for $G$ and each vertex $v \in V \setminus S$ is dominated by at most two vertices of $S$, i.e. $1 \leq \vert N(v) \cap S \vert \leq 2$.
		Moreover a set $S \subseteq V$ is a total $[1,2]$-set if for each vertex of $V$, it is the case that $1 \leq \vert N(v) \cap S \vert \leq 2$.
 The $[1,2]$-domination number of $G$, denoted $\gamma_{[1,2]}(G)$,  is the minimum number of vertices in a $[1,2]$-set. Every $[1,2]$-set with cardinality of $\gamma_{[1,2]}(G)$ is called   a $\gamma_{[1,2]}$-set.
 Total $[1,2]$-domination number and $\gamma_{t[1,2]}$-sets of $G$ are defined in a similar way.
 This paper presents a linear time algorithm to find a $\gamma_{[1,2]}$-set and  a $\gamma_{t[1,2]}$-set in generalized series-parallel graphs.

	\end{abstract}

\section{Introduction}
There is a rich theoretical literature around many important topics in graph theory, however, due to their vast new applications, we are in need of efficient algorithms for computing them. Unfortunately, most of these problems are $\mathbf{NP}$-hard. Traditionally there are two approaches to handle such problems. One is to use algorithms that yield to approximate solution\cite{guha1998approximation,guha1999improved,zou2011new}. The other is to restrict the problem to some special cases and solve the problem efficiently\cite{yen1994linear,chang2015weighted,colbourn1987locating,pfaff1984linear}.

 Finding a minimum dominating set for graphs is such a problem, which is known to be  $\mathbf{NP}$-complete. However, it is solvable in polynomial time for trees and series parallel graphs \cite{aho1974design, garey1976some,cockayne1977towards,hedetniemi1986linear}.

Also, finding a $[1,k]$-set for graphs is an  $\mathbf{NP}$-complete problem \cite{chellali20131}. In \cite{chellali20131}, it is shown that the problem of finding a $[1,2]$-set of cardinality at most $\ell$, for a given integer $\ell$ is    $\mathbf{NP}$-complete, too.

The concept of total $[1,k]$-set is proposed in \cite{chellali20131}. In \cite{total2014}, it has been proved that the problem of checking whether a given graph $G$ have any total $[1,2]$-set is  $\mathbf{NP}$-complete. For some special families of trees,  Yang and Wu computed $[1,2]$-domination number in \cite{yang20141}. Then, Goharshady et al. presented a linear algorithms to compute the minimum cardinality of $[1,2]$-sets and total $[1,2]$-sets in general trees \cite{goharshadi2015}.

 In this paper, we consider the problem of computing the minimum cardinality of $[1,2]$-sets and total $[1,2]$-sets in generalized series-parallel graphs, or GSPs for short. These graphs belong to the class of decomposable graphs which means that they can be represented by their parse trees. They can also be obtained from a set of single-edges by recursively applying the operations of series, parallel and generalized series.
 GSP graphs includes series-parallel (SP) graphs, outerplanar graph, trees, unicyclic graphs, $C_N$-trees, $C$-trees, 2-trees, cacti and filaments \cite{hare1987linear,takamizawa1982linear} (see Figure \ref{fig-gsp}).
 \begin{figure}[!h]
\centering
\includegraphics[width=12cm]{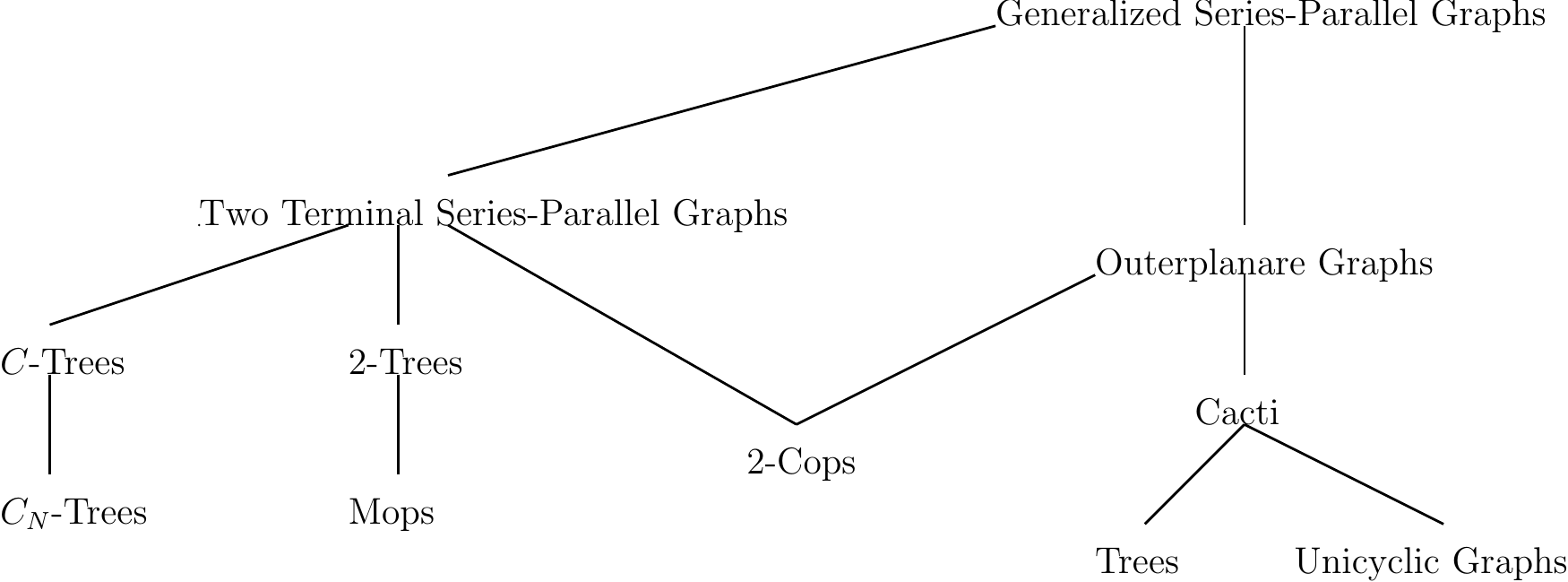}
\caption{Subfamilies of the generalized series-parallel family \cite{hare1987linear}.}\label{fig-gsp}
\end{figure}
Given a  parse tree of a GSP graph, there exist linear-time algorithms for solving many graph theoretic problems such as finding the min cut and the maximum cardinality minimal dominating set \cite{hare1987linear} and minimum dominating set for SPs \cite{hare1987linear}, to name a few see \cite{richey1988minimum}.
\\
 The main contribution of this paper are presenting  polynomial time algorithms to find $[1,2]$-sets and total $[1,2]$-sets with minimum cardinality for GSP graphs.
\\
The paper is organized as follows:
 In Section \ref{Terminology}, we study some notation and terminology of dominations and total dominations.
 In Section \ref{spgraph},  we describe generalized series-parallel graphs and corresponding parse tree.
 In Section \ref{algorithm}, we present a linear time algorithm to compute a $[1,2]$-set with minimum cardinality for generalized series-parallel graphs. In addition in this section we analyze correctness and complexity of the algorithms.
  In Section \ref{algorithm2}, we change procedures of the algorithm  to find minimum total $[1,2]$-sets  in generalized series-parallel graphs.
%that finding a $1$-dependent total $[1,2]$-set in a arbitrary graph is  $\mathbf{NP}$-complete problem and by using these results
%we will prove investigating equality of $[1,2]$-domination number with the number of vertices of a graph is also  $\mathbf{NP}$-complete.
%moreover we will study complexity of problems that check equality of total $[1,2]$-domination number and
%total domination number and equality of  $[1,2]$-domination number and  domination number for arbitrary graphs.
\\
\section{Terminology and notation} \label{Terminology}
In this section, we review some required terminology and notations of graph theory. For other notation which is not defined  here, the reader is refer to \cite{west2001introduction}.

Let  $G=(V,E)$ be a simple graph.
The open and closed neighborhoods of a vertex $v \in V(G)$ are defined as
$$N_G(v)=\{u : uv \in E(G)\},$$
and
$$N_G[v]=N_G(v) \cup \{v\},$$
respectively. Note that we omit the subscript $G$, whenever there is no risk of confusion.

  A subset  of vertices of  $G$  such as $D$ is called a dominating set, if every $v \in V$ is either an element of $D$ or is adjacent to an element of $D$. In other words, for any $v \in V\backslash D$, it is the case that $\vert N(v) \cap D \vert \geq 1$. Similarly, a set $D$ of vertices of $G$ is called a total dominating set, if every $v \in V$  is adjacent to an element of $D$, i.e. $\vert N(v) \cap D \vert \geq 1$ holds for all $v \in V$.
A $\gamma$-set is a dominating set of $G$ with minimum cardinality. The size of such a set is called domination number of $G$ and  is denoted by $\gamma(G)$.
   Similarly, a $\gamma_t$-set is a total dominating set of $G$ with minimum cardinality and total domination number of $G$  is the cardinality of such a set which is denoted by  $\gamma_{t}(G)$.
   
A set $S \subseteq V$ is called a $[1, 2]$-set of $G$ if for each $v \in V \setminus S$ we have $1 \leq \vert N(v) \cap S \vert \leq 2$, i.e. $v$ is adjacent to at least one but not more than two vertices in $S$. The size of the smallest $[1, 2]$-sets of $G$ is denoted by $\gamma_{[1, 2]} (G)$. Any such set is called a $\gamma_{[1, 2]}$-set of $G$.

 As a generalization of $[1, 2]$-sets, for two nonnegative integers  $j$ and $k$,  where $j\leq k$, a set $S \subseteq V$ is an $[j, k]$-set if for each $v \in V \setminus S$ we have $j \leq \vert N(v) \cap S \vert \leq k$. A set $S' \subseteq V$ is called a total $[j, k]$-set of $G$ if for each $v \in V$ we have $j \leq \vert N(v) \cap S \vert \leq k$, i.e. if each vertex, no matter in $S'$ or not, has at least $j$ and at most $k$ neighbors in $S'$. These definitions can be found in \cite{chellali20131}.

\section{Generalized Series-Parallel Graphs and Parse Tree} \label{spgraph}
First we define the class of graphs that we study in this paper. In addition, we review some known results about them.

\begin{definition}[Generalized Series-Parallel Graphs]\cite{chebolu2012exact}
A generalized series-parallel (GSP) graph is any graph $G=(V,E,s,t)$ with two distinguished nodes $s,t\in V$, called terminals, which is defined recursively as follows:
\begin{itemize}
\item[$(1)\;\;o_i$:]
The graph $G$ consisting of two vertices connected by a single edge is a GSP graph.
\item[$(2)\;\;o_s$:]
Given two GSP graphs $G_1=(V_1,E_1,s_1,t_1)$, $G_2=(V_2,E_2,s_2,t_2)$, the series operation of $G_1$ and $G_2$  creates a new GSP graph $G=G_1 o_s G_2=(V,E, s_1,t_2)$, where
$$V=V_1\cup V_2\setminus \{s_2\},$$
and
$$E=E_1 \cup E_2 \cup\{t_1v: v\in N_{G_2}(s_2)\}\setminus \{s_2v: v\in N_{G_2}(s_2)\}.$$

\item[$(3)\;\;o_p$:]
Given two GSP graphs $G_1=(V_1,E_1,s_1,t_1)$, $G_2=(V_2,E_2,s_2,t_2)$, the parallel operation of $G_1$ and $G_2$  creates a new GSP graph  $G=G_1 o_p G_2=(V,E, s_1,t_1)$, where
$$V=V_1\cup V_2\setminus \{s_2,t_2\}, and$$
$$E= E_1 \cup E_2 \cup\{s_1 v: v\in N_{G_2}(s_2)\}\cup\{t_1v: v\in N_{G_2}(t_2)\}\setminus
    (\{s_2 v: v\in N_{G_2}(s_2)\} \cup \{t_2v: v\in N_{G_2}(t_2)\})$$

\item[$(4)\;\;o_g$:]
Given two GSP graphs $G_1=(V_1,E_1,s_1,t_1)$, $G_2=(V_2,E_2,s_2,t_2)$, the generalized series operation of $G_1$ and $G_2$  creates a new GSP graph $G=G_1 o_g G_2=(V,E, s_1,t_1)$, where
$$V=V_1\cup V_2\setminus \{s_2\}, and$$
$$E=E_1 \cup E_2 \cup\{t_1v: v\in N_{G_2}(s_2)\}\setminus \{s_2v: v\in N_{G_2}(s_2)\}.$$
\item[$(5)\;\;\;\;\;\;$]
Any GSP graph is obtained by finite application of rules $(1)$  through $(4)$.
\end{itemize}
  \end{definition}
  Figure \ref{fig-gspop} illustrates the example of these rules  to construct a GSP graph.
\begin{figure}[!h]
\centering
\includegraphics[width=12cm]{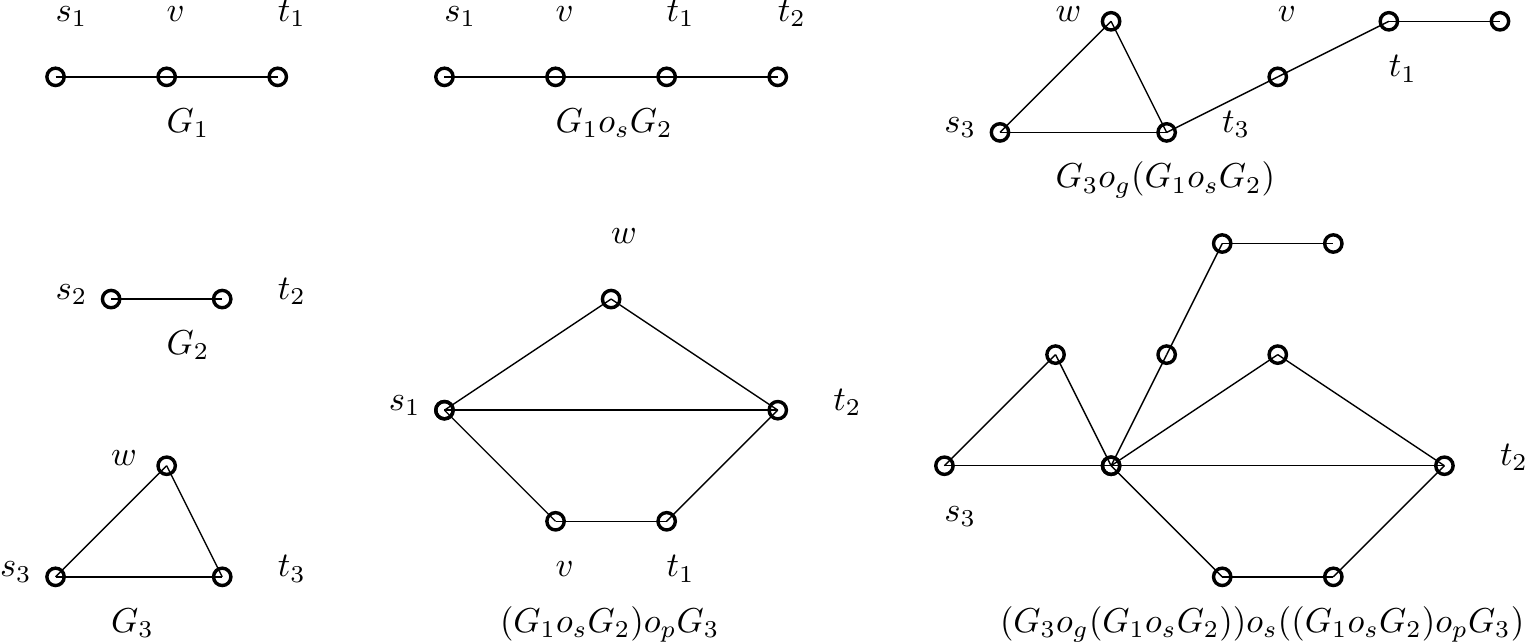}
\caption{Example of operations for constructing GSP graphs}\label{fig-gspop}
\end{figure}

Moreover, the subclass of GSPs obtained by finite application of rules $(1)$ to $(3)$ is called series-parallel or SP class.
 Note that in Figure \ref{fig-gspop} the final  graph $(\hat{G}o_g(G_1o_sG_2))o_s((G_1o_sG_2)o_p\hat{G})$ is a GSP graph but it is not SP.

$p$-graphs are defined as an special  kind of subgraphs of GSP graph which are reviewed next.
\begin{definition}[$p$-graph]

Let $G = (V, E,x,y)$  be a GSP and  $\hat{G} =(\hat{V}, \hat{E},\hat{x},\hat{y})$ subgraph of $G$ satisfy  the following conditions:
\begin{enumerate}
\item
 $\hat{x}=x$ or there exists an edge $ \{u,v\} \in E \setminus \hat{E}$ such that $x \notin \hat{V}$  and $v = \hat{x}\in \hat{V}$.
\item
 $\hat{y}=y$ or there exists an edge $ \{w,z\} \in E \setminus \hat{E}$ such that $w \notin \hat{V}$  and $z = \hat{y}\in \hat{V}$.
\end{enumerate}
then $\hat{G}$ is called a $p$-graph of $G$.
\end{definition}
\begin{remark}
Note that the concept of $p$-graph is defined for SPs in  \cite{kikuno1983linear}, we have extended its definition to GSPs.
\end{remark}
 A generalized series-parallel graph $G$ can be represented by a binary parse tree $T$ which is defined as follows.

 \begin{definition}[Binary Parse Tree for Generalized Series-Parallel Graphs]\cite{kikuno1983linear}
 The binary parse tree of GSP $G$  is defined recursively as follows:
 \begin{enumerate}
 \item
 A tree consisting of a single vertex labeled $(u,v)_i$ is a binary parse tree of primitive  GSP  $G=(\{u,v\},\{\{u,v\}\})$.
 \item
 Let $G=(V,E)$ be a GSP by some composition of two other GSPs $G_1$ and $G_2$, and  $T_1$ and $T_2$ be the binary parse trees of them, respectively. Then, the binary parse tree of $G$ is a tree with the root $r$ labeled as either  $(u,v)_s$, $(u,v)_p$ or $(u,v)_g$ depending on which operation is applied to generate $G$. Vertices $u$ and $v$ are terminals of $G$ and roots of $T_1$ and $T_2$ are the left and right children of $r$, respectively.

 Obviously, for any binary parse tree of GSP $G$, every internal vertex of $G$  has exactly two children and there are $\vert E \vert$  leaves.

Note that when we use a label $(x,y)$ , we do not care about the label being either $i,s,p$ and $g$.

Let $t$ be an internal vertex of $T$ and $\tau(t)$ denote the subtree of $T$ rooted at $t$ . Also the left and right subtree of $t$  are denoted as  $\tau_l(t)$  and  $\tau_r(t)$,  respectively. Then the vertices of $T$  are labeled as follows:

 \begin{enumerate}
 \item
 For each edge $e=\{x,y\} \in E$, there exists exactly one leaf which is labeled by $(x,y)$ in $T$.
 \item
 For each internal vertex $t \in V_T$   labeled as $(x,y)_s$, the root of $\tau_l(t)$  is labeled as $(x,z)$ and the root of $\tau_r(t)$  is labeled as $(z,y)$, where $z$ is some vertex of $V$. These vertices are called $s$-vertices.
 \item
 For each internal vertex $t \in V_T$  is labeled as $(x,y)_p$, the root of $\tau_l(t)$ and $\tau_r(t)$  are labeled as $(x,y)$. These vertices are called $p$-vertices.
 \item
 For each internal vertex $t \in V_T$ labeled $(x,z)_g$, the root of $\tau_l(t)$  is labeled as $(x,z)$,  and the root of $\tau_r(t)$ is labeled as $(z,y)$, where $z$ is a vertex of $V$. These vertices are called $g$-vertices.
 \end{enumerate}
 \end{enumerate}

\end{definition}
Figure \ref{parsetree-g} illustrates a binary parse tree for a  GSP.
\begin{figure}[!h]
\centering
\includegraphics[width=16cm]{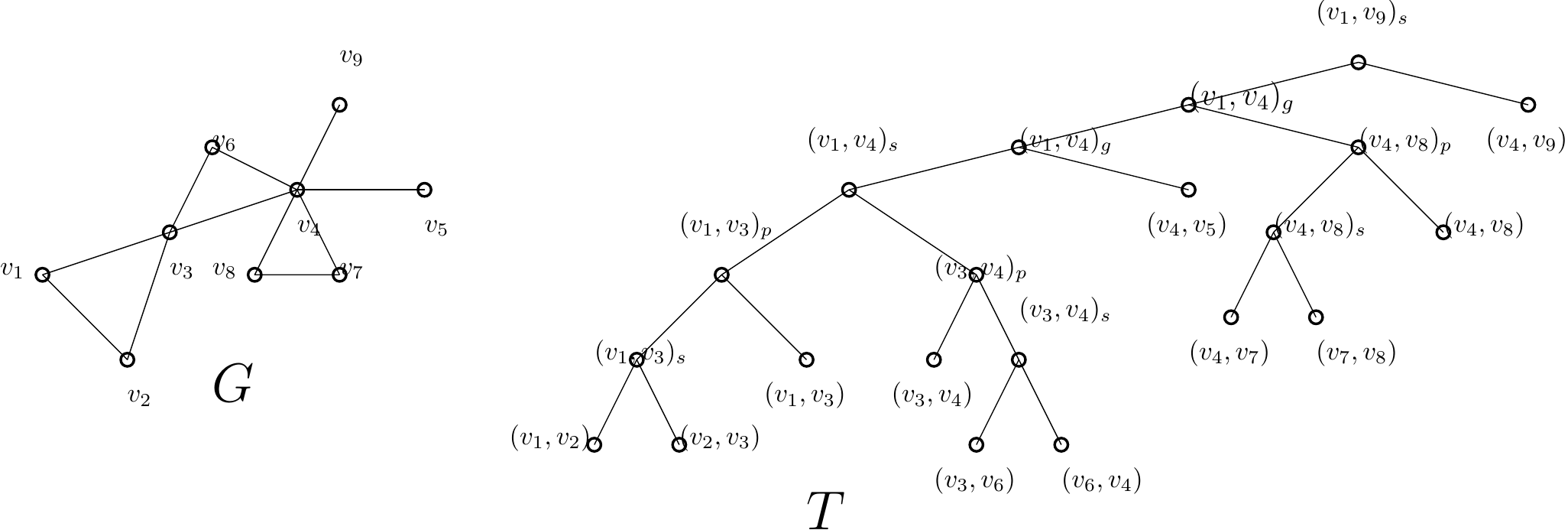}
\caption{ A GSP graph and a binary parse tree of it.}\label{parsetree-g}
\end{figure}

 Note that the binary parse tree of a GSP graph is not  necessarily  unique.
 % and we do not need the uniqueness of binary parse tree in our algorithm.
 \begin{lemma}\cite{hopcroft1973dividing}
	For a given  GSP like $G$, a binary parse tree can be found in linear time.
	\end{lemma}
\section{Basic blocks of algorithms for finding minimum $[1,2]$-set}\label{algorithm}
 We are given a GSP $G$, at the first step, this algorithm find  a parse tree like $T$ of given $G$. Each subtree of parse tree is corresponding to a $p$-graph of $G$.
Four procedures ProcessLeaf, ProcessSvertex,   ProcessPvertex and ProcessGvertex form the basic block of algorithm.
Our goal is to find a $\gamma_{[1,2]}$-set for a given graph $G$.

\subsection*{Definitions}
Let $t$ is a vertex of $T$ and $\hat{G}$ is a $p$-graph to subtree with root $t$ of $T$. We define the following:
\begin{itemize}
\item[$\bullet$]
For any vertex $v$ of $T$, the set $ch(v)$ consist of all children of $v$. In this parse tree, if $v$ is a leaf,  $ch(v)$ will be an empty set and if it is internal vertex, $ch(v)$ will contain two elements.
\item[$\bullet$]
Let $(x,y)$ be the label of $v$, $S(x_{i,j},y_{i',j'})$ is a subset of $V(\hat{G})$ like $D$, such that the number  vertices of $V(\hat{G})\cap D$ and $V(G)\setminus V(\hat{G})\cap D$ which  dominate $x$ be $i$  and $j$, respectively. Similarly  $i'$ vertex of $V(\hat{G})\cap D$ and $j'$ vertex of $V(G)\setminus V(\hat{G})\cap D$ dominate $y$. In this algorithm if $i=j=0$ then $x \in D$ and if $i'=j'=0$ then $y \in D$.
In a similar way, we can define $S_\ell(x_{i_\ell,j_\ell},y_{i'_\ell,j'_\ell})$ and $S_r(x_{i_r,j_r},y_{i'_r,j'_r})$ for the root of left and right subtree of $v$ respectively.
\item[$\bullet$]
The function $Minsize$ received a number of sets as input and return a set among them with minimum cardinality.
\end{itemize}

This algorithm in each step of traversing binary parse tree, for each visiting vertex $v$ of $T$ recall one of the  procedures ProcessLeaf, ProcessSvertex,   ProcessPvertex and ProcessGvertex to find a subset of $S(x_{i,j},y_{i',j'})\subseteq V(\hat{G})$ such that junction of $S(x_{i,j},y_{i',j'})$ and a subset $V(G)\setminus V(\hat{G})$ form a $\gamma_{[1,2]}$-set for a $p$-graph $\hat{G}$ of $G$.
 So the following cases can not occur for every vertex $x\in V(\hat{G})$.
\begin{itemize}
\item[1.]
$x$ is dominated by one vertex of $V(\hat{G})$ and two vertices of $V(G)\setminus V(\hat{G})$.
\item[2.]
$x$ is dominated by two vertices of $V(\hat{G})$ and one vertex of $V(G)\setminus V(\hat{G})$.
\item[3.]
$x$ is dominated by two vertices of $V(\hat{G})$ and two vertices of $V(G)\setminus V(\hat{G})$.
\end{itemize}
So  $S(x_{i,j},y_{i',j'})$, will be computed for $(i,j),(i',j')\in M$ such that $M=\{(0,0),(0,1),(0,2),(1,0)(1,1),(2,0)\}$.

\subsection*{Procedure for Leaves of $T$}
Input of this procedure is a leaf of $v\in V_T$  labeled by $(x,y)_i$ and output $S(x_{i,j},y_{i',j'})$ for the leaf labeled by $(x,y)_i$ and for all $(i,j),(i',j' )\in M$.
We compute $S(x_{i,j},y_{i',j'})$ for leaves using  table \ref{table1}.
 \begin{algorithm}[h!]
\begin{algorithmic}[1]
\Procedure{ProcessLeaf}{$x,y$}
\State $S(x_{0,0},y_{0,0})\gets \{x,y\} $
                    \State $S(x_{0,0},y_{1,0})\gets \{x\} $
                    \State $S(x_{0,0},y_{1,1})\gets \{x\} $
                    \State $S(x_{1,0},y_{0,0})\gets \{y\} $
                    \State $S(x_{1,1},y_{0,0})\gets \{y\} $
		            \ForAll{$j=0,1,2\;and\;j'=1,2$ }
                    \State $S(x_{0,j},y_{0,j'})\gets \emptyset $
                    \EndFor
                    \State $S(x_{0,1},y_{0,0})\gets \emptyset $
                    \State $S(x_{0,2},y_{0,0})\gets \emptyset $
                    \ForAll{$(i,j)\in M$}
                    \State $S(x_{i,j},y_{2,0})\gets NaN$
                    \State $S(x_{2,0},y_{i,j})\gets NaN$
                    \EndFor
                    \ForAll{$(i,j)\in M \setminus \{(0,0),(2,0)\}$}
                    \State $S(x_{i,j},y_{1,0})\gets NaN $
                    \State $S(x_{i,j},y_{1,1})\gets NaN $
                    \State $S(x_{1,0},y_{i,j})\gets NaN $
                    \State $S(x_{1,1},y_{i,j})\gets NaN $
                    \EndFor
\EndProcedure
	\end{algorithmic}
\end{algorithm}

 \begin{table}[h]
\begin{center}
\caption{$S(x_{i,j},y_{i',j'})$ for different values of $(i,j)$ in rows and $(i',j')$ in columns}\label{table1}
\begin{tabular}{lllllll}
\hline
  &$(0,0)$ & $(0,1)$ & $(0,2)$ & $(1,0)$ & $(1,1)$ & $(2,0)$\\
\hline

  $(0,0)$ & $\{x,y\}$ & $\emptyset$ & $\emptyset$ & $\{y\}$ & $\{y\}$ & NaN\\
  $(0,1)$ & $\emptyset$ & $\emptyset$ & $\emptyset$ & NaN & NaN & NaN\\
  $(0,2)$ & $\emptyset$ & $\emptyset$ & $\emptyset$ & NaN & NaN & NaN\\
  $(1,0)$ & $\{x\}$ & NaN & NaN & NaN & NaN & NaN\\
  $(1,1)$ & $\{x\}$ & NaN & NaN & NaN & NaN & NaN\\
  $(2,0)$ & NaN & NaN & NaN & NaN & NaN & NaN\\

\end{tabular}

\end{center}
\end{table}

A leaf of tree which is labeled by $(x,y)$  is corresponding to an edge $\{x,y\}$  of GSP graph $G$. Following cases will occur to compute $S(x_{i,j},y_{i',j'})$.

\begin{itemize}
 \item
 $i=j=0$ and $i'=j'= 0$: $x,y \in D$ and  $S(x_{i,j},y_{i',j'})=\{x,y\}$.
 \item
$i=1$, $i'=j'= 0$: $x$ is dominated by $y$ and  $S(x_{i,j},y_{i',j'})=\{y\}$.
 \item
$i'=1$, $i=j= 0$: $y$ is dominated by $x$ and  $S(x_{i,j},y_{i',j'})=\{x\}$.
\item
 $i=i'=0$,  $j\neq 0$ and $j'\neq 0$: $x$ and $y$ is dominated by one vertex or two vertices  of vertex $V(G)\setminus V(\hat{G})$, so $x,y \notin D$ and  $S(x_{i,j},y_{i',j'})$ is an empty set.
\item
$i=0$ and $j \in \{1,2\}$: It implies that $x\notin D$  and it is dominated by some vertices $V(G)\setminus V(\hat{G})$, so $y\notin D$ and $S(x_{i,j},y_{0,0})$ is undefinable. Similarly If $i'=0$,  $j' \in \{1,2\}$ and $i=j=0$ $S(x_{i,j},y_{i',j'})$ is undefinable.
\item
$i=1$, $i'\neq 0$ or $j'\neq 0$:  $S(x_{i,j},y_{i',j'})$ is undefinable.  $i=1$ implies that $x$ is dominated by exactly one vertex of $\hat{G}$, this vertex is $y$. So $y\in D$ and $i'=j'=0$, similarly for $i'=1$, if $i\neq 0$ or $j\neq 0$,  $S(x_{i,j},y_{i',j'})$ is undefinable.
\item
$i=2$, since there is exactly one vertex set of $\hat{G}$ to dominating $x$, $S(x_{i,j},y_{i',j'})$ is undefinable and similarly for $i'=2$,  $S(x_{i,j},y_{i',j'})$ is undefinable.

\end{itemize}

\subsection*{Sets for $s$-vertices of $T$}
Let $t$ is a vertex of $T$ is labeled by $(x,y)_s$. In this procedure, the set $S(x_{i,j},y_{i',j'})$ will computed for a given vertex $x$, $y$ and common vertex $z$. Assume the sets corresponding to  $\tau_l(t)$ and $\tau_r(t)$ are $S_\ell(x_{i_\ell,j_\ell},z_{i'_\ell,j'_\ell})$ and $S_r(z_{i_r,j_r},y_{i'_r,j'_r})$ respectively.

\begin{algorithm}[h]
\begin{algorithmic}[1]
\Procedure{ProcessSvertex}{$x,z,y$}
 \ForAll{$(i,j),(i',j')\in M$}
                    \State $setlist\gets \emptyset$
                    \ForAll{$(i'_\ell,j'_\ell)\in M$}
                    \State Add $S_\ell(x_{i,j},z_{i'_\ell,j'_\ell})\cup S_r(z_{j'_\ell,i'_\ell},y_{i',j'})$ to $setlist$
                    \EndFor
                    \State $S(x_{i,j},y_{i',j'})\gets Minsize(setlist)$
                    \EndFor
\EndProcedure
	\end{algorithmic}
\end{algorithm}

Let root of $\tau_l(t)$ and $\tau_r(t)$ labeled by $(x,z)$ and $(z,y)$ respectively, for some $z \in V$.
Let $S_\ell(x_{i_\ell,j_\ell},z_{i'_\ell,j'_\ell})$ and $S_r(z_{i_r,j_r},y_{i'_r,j'_r})$ be sets associated with the vertex $(x,z)$ and $(z,y)$ of $T$.
Now for each $s$-vertex of $T$, we can compute sets $S(x_{i,j},y_{i',j'})$ for all $(i,j),(i',j') \in M$ as follow,
\begin{equation} \label{svertex}
S(x_{i,j},y_{i',j'})=Minsize_{(i,j),(i',j')\in M}\{S_\ell(x_{i,j},z_{i'_\ell,j'_\ell}) \cup s_r(x_{j'_\ell,i'_\ell},y_{i',j'})\}
\end{equation}
To prove the correctness of formula \ref{svertex}, let $\tau_l(t)$, $\tau_r(t)$ and $\tau(t)$ represent vertices of $T$ which is corresponding  $p$-graphs $G_1=(V_1,E_1,x,z)$, $G_2=(V_2,E_2,z,y)$ and $\hat{G}=G_1o_s G_2=(\hat{V},\hat{E},x,y)$ respectively. Figures  \ref{smethod1}-\ref{smethod0} show the computation $S(x_{i,j},y_{i',j'})$ when $z$ is dominated by vertices of $G_1$ or vertices of $G_2$. According to the number of vertices in $D \cap V(G_1)$, $D \cap (V(G)\setminus V(G_1))$, $D \cap V(G_2)$ and $D\cap (V(G)\setminus V(G_2))$ that dominate $z$, the following cases will occur.
\begin{itemize}
\item[$\bullet$]
$z \in D$, in this case ${(0,0)\in M}$, see Figure \ref{smethod0}.

\begin{figure}[!h]
\centering
\includegraphics[width=6cm]{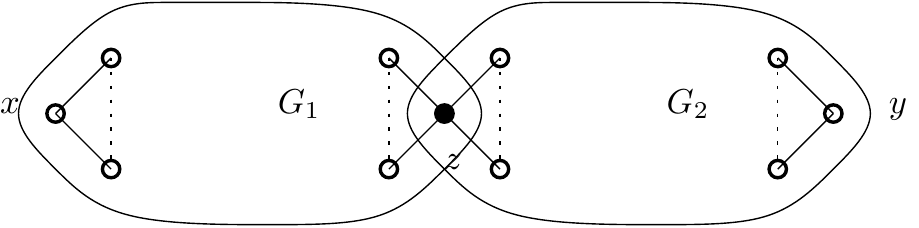}
\caption{ $z \in D$.}\label{smethod0}
\end{figure}

\item[$\bullet$]
$z\notin D$, if $z$ is dominated by  one vertex of $G_1$ then $(0,1)\in M$, see Figure \ref{smethod1}.a or if it is dominated by  one vertex of $G_2$ then then $(1,0)\in M$, see Figure \ref{smethod1}.b.

\begin{figure}[!h]
\centering
\includegraphics[width=6cm]{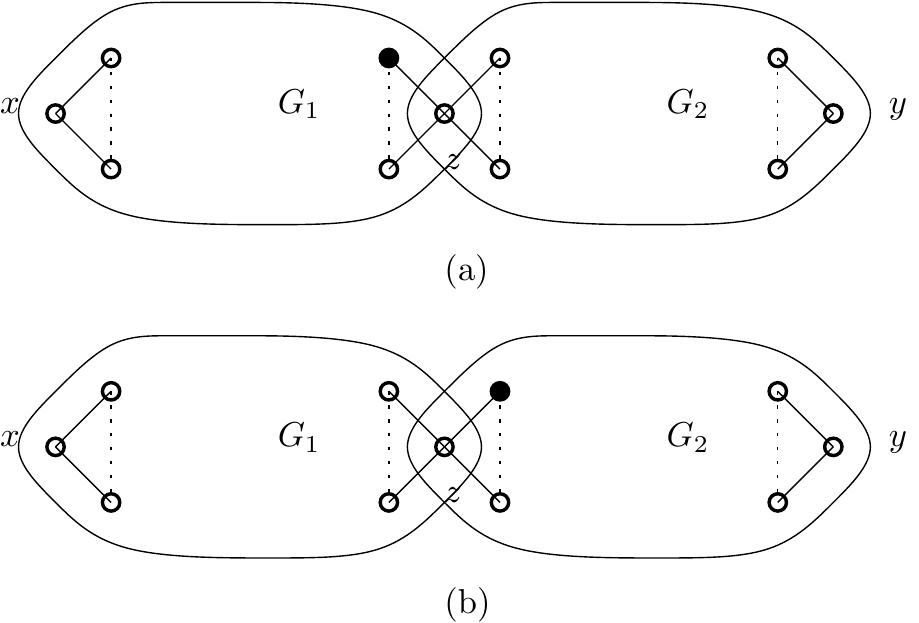}
\caption{$z$ is dominated  once.}\label{smethod1}
\end{figure}

\item[$\bullet$]
$z\notin D$, if it is dominated by  two vertices of $G_1$ then ${(0,2)\in M}$, see Figure \ref{smethod2}.a or if it is dominated by  two vertices of $G_2$ then ${(2,0)\in M}$, see Figure\ref{smethod2}.b.

\begin{figure}[!h]
\centering
\includegraphics[width=6cm]{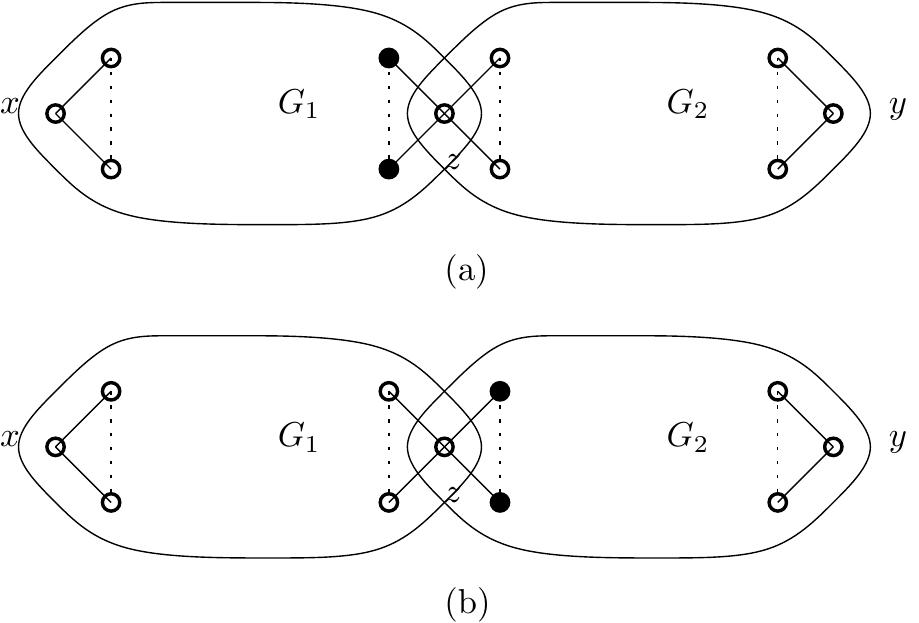}
\caption{$z$ is dominated by two vertices of $V(G_1)$ or $V(G_2)$.}\label{smethod2}
\end{figure}
\item[$\bullet$]
$z\notin D$, if it is dominated by  one vertex of $G_1$ and  it is dominated by  one vertex of $G_2$ then ${(1,1)\in M}$, see Figure\ref{smethod3}.

\begin{figure}[!h]
\centering
\includegraphics[width=6cm]{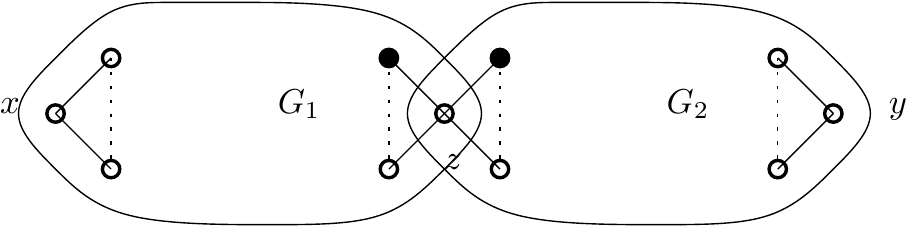}
\caption{$z$ is dominated by one vertex of $V(G_1)$ and the other of  $V(G_2)$.}\label{smethod3}
\end{figure}
\end{itemize}

\subsection*{Sets for $p$-vertices of $T$}
Let $t$ is a vertex of $T$ is labeled by $(x,y)_p$, in this procedure, the set $S(x_{i,j},y_{i',j'})$ will computed for a given vertex $x$, $y$. The sets corresponding to  $\tau_l(t)$ and $\tau_r(t)$ are $S_\ell(x_{i_\ell,j_\ell},y_{i'_\ell,j'_\ell})$ and $S_r(x_{i_r,j_r},y_{i'_r,j'_r})$ respectively.

\begin{algorithm}[h!]
\begin{algorithmic}[1]
\Procedure{ProcessPvertex}{$x,y$}
 \ForAll{$(i,j),(i',j')\in M$}
                    \State $setlist[(i,j),(i',j')]\gets \emptyset$
                    \EndFor
      \ForAll{$(i,j)\in \{(0,0),(0,1),(0,2)\}$}
                    \ForAll{$(i',j')\in \{(0,0),(0,1),(0,2)\}$}
                    \State $setlist[(i,j),(i',j')]\gets S_\ell(x_{i,j},y_{i',j'})\cup S_r(x_{i,j},y_{i',j'})$
                    \EndFor

                    \State $
                    \begin{aligned}
                    Setlist[(i,j),(1,0)]\gets &\{S_\ell(x_{i,j},y_{1,0})\cup S_r(x_{i,j},y_{0,1}),S_\ell(x_{i,j},y_{0,1})\cup S_r(x_{i,j},y_{1,0})\}
                    \end{aligned}$

                    \State $
                    \begin{aligned}
                    Setlist[(i,j),(1,1)]\gets &\{S_\ell(x_{i,j},y_{1,1})\cup S_r(x_{i,j},y_{0,2}),S_\ell(x_{i,j},y_{0,2})\cup S_r(x_{i,j},y_{1,1}) \}
                    \end{aligned}$

                    \State $
                    \begin{aligned}
                    Setlist[(i,j),(2,0)]\gets & \{S_\ell(x_{i,j},y_{2,0})\cup S_r(x_{i,j},y_{0,2}),S_\ell(x_{i,j},y_{0,2})\cup S_r(x_{i,j},y_{2,0}),\\
                    &S_\ell(x_{i,j},y_{1,1})\cup S_r(x_{i,j},y_{1,1}) \}
                    \end{aligned}$

                   \State $
                    \begin{aligned}
                    Setlist[(1,0),(i,j)]\gets & \{S_\ell(x_{1,0},y_{i,j})\cup S_r(x_{0,1},y_{i,j}),S_\ell(x_{0,1},y_{i,j})\cup S_r(x_{1,0},y_{i,j}) \}
                     \end{aligned}$

                    \State $
                     \begin{aligned}
                     Setlist[(1,1),(i,j)]\gets &\{S_\ell(x_{1,1},y_{i,j})\cup S_r(x_{0,2},y_{i,j}),S_\ell(x_{0,2},y_{i,j})\cup S_r(x_{1,1},y_{i,j}) \}
                       \end{aligned}$

                    \State $
                    \begin{aligned}
                    Setlist[(2,0),(i,j)]\gets & \{S_\ell(x_{2,0},y_{i,j})\cup S_r(x_{0,2},y_{i,j}),S_\ell(x_{0,2},y_{i,j})\cup S_r(x_{2,0},y_{i,j}),\\
                    & S_\ell(x_{1,1},y_{i,j})\cup S_r(x_{1,1},y_{i,j})\}
                     \end{aligned}$
                    \EndFor
    	
                    \ForAll{$(i,j),(i',j')\in \{(0,1),(1,0)$}
                    \State Add $S_\ell(x_{i,j},y_{i',j'})\cup S_r(x_{j,i},y_{j',i'})$ to $setlist[(i,j),(i',j')]$
                    \EndFor

                    \ForAll{$(i,j)\in \{(0,1),(1,0)$}
                    \State Add $S_\ell(x_{i,j},y_{1,1})\cup S_r(x_{j,i},y_{0,2})$ to $setlist[(i,j),(1,1)]$
                    \State Add $S_\ell(x_{i,j},y_{0,2})\cup S_r(x_{j,i},y_{1,1})$ to $setlist[(i,j),(1,1)]$

                    \State Add $S_\ell(x_{i,j},y_{0,2})\cup S_r(x_{j,i},y_{2,0})$ to $setlist[(1,0),(2,0)]$
                    \State Add $S_\ell(x_{i,j},y_{2,0})\cup S_r(x_{j,i},y_{0,2})$ to $setlist[(1,0),(2,0)]$
                    \State Add $S_\ell(x_{i,j},y_{1,1})\cup S_r(x_{j,i},y_{1,1})$ to $setlist[(1,0),(2,0)]$

                    \State Add $S_\ell(x_{1,1},y_{i,j})\cup S_r(x_{0,2},y_{j,i})$ to $setlist[(1,1),(1,0)]$
                    \State Add $S_\ell(x_{0,2},y_{i,j})\cup S_r(x_{1,1},y_{j,i})$ to $setlist[(1,1),(1,0)]$
                    \EndFor
                    \State
                        $
                        \begin{aligned}
                            setlist[(1,1),(1,1)] \gets & \{S_\ell(x_{1,1},y_{1,1})\cup S_r(x_{0,2},y_{0,2}),S_\ell(x_{1,1},y_{0,2})\cup S_r(x_{0,2},y_{1,1}),\\
                             & S_\ell(x_{0,2},y_{1,1})\cup S_r(x_{1,1},y_{0,2}),S_\ell(x_{0,2},y_{0,2})\cup S_r(x_{1,1},y_{1,1})\}
                        \end{aligned}
                        $
 \algstore{myalg}
	\end{algorithmic}
\end{algorithm}

\begin{algorithm}
	\begin{algorithmic}
		\algrestore{myalg}

                    \ForAll{$(i',j')\in \{(0,2),(2,0),(1,1)\}$}
                    \State Add $S_\ell(x_{1,1},y_{i',j'})\cup S_r(x_{0,2},y_{j',i'})$ to $setlist[(1,1),(2,0)]$
                    \State Add $S_\ell(x_{1,1},y_{i',j'})\cup S_r(x_{0,2},y_{j',i'})$ to $setlist[(1,1),(2,0)]$
                    \EndFor

                    \ForAll{$(i,j)\in \{(0,2),(2,0),(1,1)\}$}
                    \State Add $S_\ell(x_{i,j},y_{1,1})\cup S_r(x_{j,i},y_{0,2})$ to $setlist[(2,0),(1,1)]$
                    \State Add $S_\ell(x_{i,j},y_{0,2})\cup S_r(x_{j,i},y_{1,1})$ to $setlist[(2,0),(1,1)]$
                    \EndFor

                    \ForAll{$,(i,j)(i',j')\in \{(0,2),(2,0),(1,1)\}$}
                    \State Add $S_\ell(x_{i,j},y_{i',j'})\cup S_r(x_{j,i},y_{j',i'})$ to $setlist[(2,0),(2,0)]$
                    \EndFor

                    \ForAll{$(i,j),(i',j')\in M$}
                    \State $S(x_{i,j},y_{i',j'})\gets Minsize(setlist[(i,j),(i',j')])$
                    \EndFor
\EndProcedure
\end{algorithmic}
\end{algorithm}
\newpage
Since the number of vertices of $V(G_1)$ and $V(G_2)$ dominate  $x$ make changes in value of  $i$ and similarly the number of vertices of $V \setminus V(G_1)$ and $V \setminus V(G_2)$ make changes in value of  $j$, for each $(i,j) \in M$ we describe the method for finding $S(x_{i,j},y_{i',j'})$. It is enough to find a relation between values of $(i,j)$,  $(i_\ell,j_\ell)$ and $(i_r,j_r)$ in formula \ref{pvertex}.
\begin{equation} \label{pvertex}
S(x_{i,j},y_{i',j'})=Minsize\{S_\ell(x_{i_\ell,j_\ell},y_{i'_\ell,j'_\ell}) \cup s_r(x_{i_r,j_r},y_{i'_r,j'_r})\}.
\end{equation}
 To find this relation, let $\tau_l(t)$, $\tau_r(t)$ and $\tau(t)$ of $T$ corresponding to $p$-graphs $G_1=(V_1,E_1,x,y)$, $G_2=(V_2,E_2,x,y)$ and $\hat{G}=G_1 o_pG_2=(\hat{V},\hat{E},x,y)$ respectively.
   Now according to $(i,j)$ (resp. $(i',j')$) values for $(i_\ell,j_\ell)$ and $(i_r,j_r)$ (resp. $(i'_\ell,j'_\ell)$ and $(i'_r,j'_r)$) are determined.
 \begin{itemize}
 \item[$\bullet$]
 $x \in D$, in formula \ref{pvertex}, if $i=j=0$ then $i_\ell=j_\ell=i_r=j_r=0$.
 We define: $$S(x_{0,0},y_{i',j'})=Minsize\{s_\ell(x_{0,0},y_{i'_\ell,j'_\ell}) \cup s_r(x_{0,0},y_{i'_r,j'_r})\}.$$
  
\begin{figure}[!h]
\centering
\includegraphics[width=6cm]{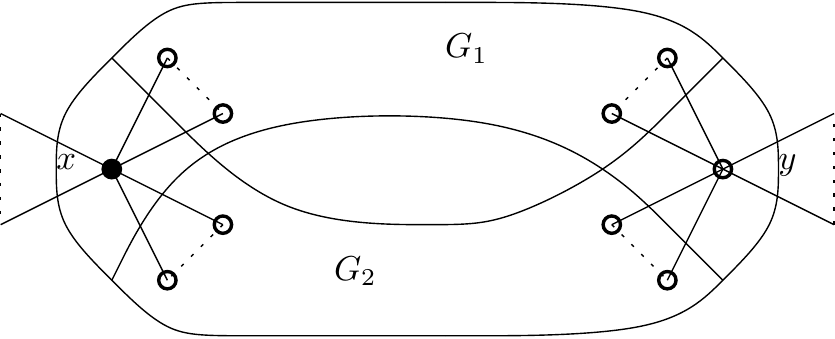}
\caption{ $x \in D$}\label{pmethod0}
\end{figure}

 \item[$\bullet$]
 $x \notin D$, if $x$ is not dominated by any vertex of $V(\hat{G})$ then it is dominated by one vertex  or two vertices of $V(G)\setminus V(\hat{G})$,   see Figure \ref{pmethod1}.
 In the left side of formula \ref{pvertex}, $i=0$ and $j \neq 0$, so  $i_\ell=i_r=0$ and $j_\ell=j_r=j$.

We define $S(x_{0,1},y_{i',j'})=Minsize_{i',j'\in\{0,1,2\}}\{S_\ell(x_{0,1},y_{i'_\ell,j'_\ell}) \cup s_r(x_{0,1},y_{i'_r,j'_r})\}$, see Figure \ref{pmethod1}.a, and
$S(x_{0,2},y_{i',j'})=Minsize\{S_\ell(x_{0,2},y_{i_\ell',j_r'}) \cup s_r(x_{0,2},y_{i_r',j_r'})\}$, see Figure \ref{pmethod1}.b.
\begin{figure}[!h]
\centering
\includegraphics[width=6cm]{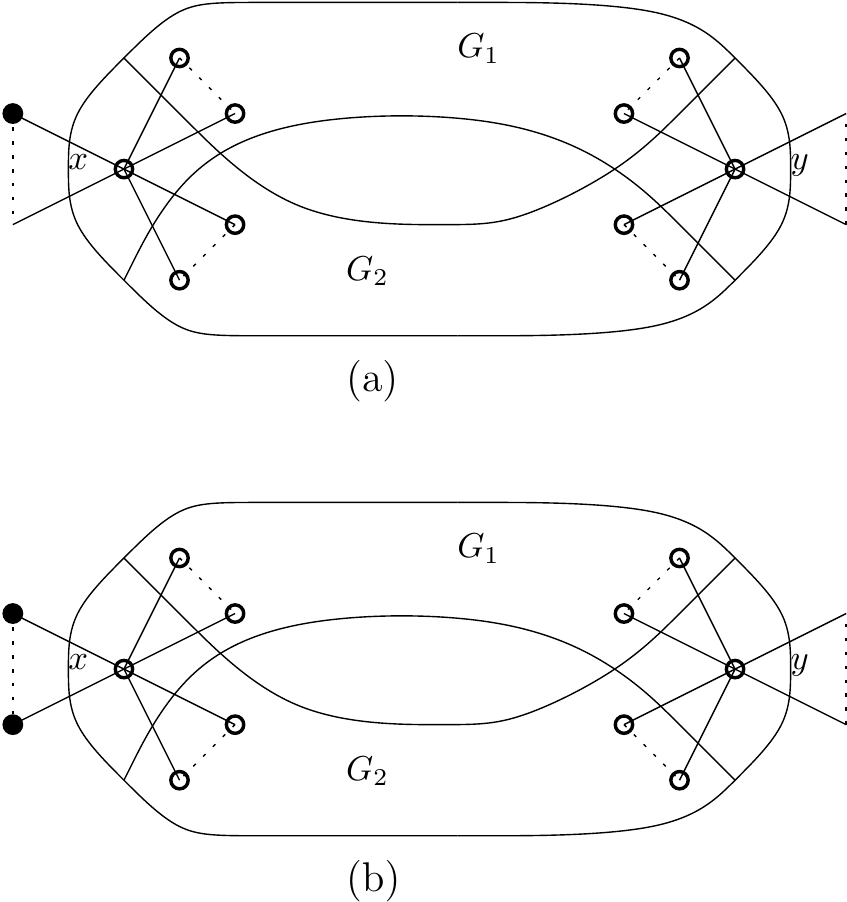}
\caption{ $x$ is not dominated one vertex  or two vertices of $V(G) \setminus V(\hat{G})$.}\label{pmethod1}
\end{figure}

\item[$\bullet$]
 $x \notin D$, if $x$ is dominated by exactly one vertex of $V(\hat{G})$ then, it is  dominated by a vertex of $G_1$ that is out of $G_2$, see Figure \ref{pmethod2}.a
 or it is dominated by  a vertex of $G_2$ that is out of $G_1$, see Figure \ref{pmethod2}.b.

 It means that in formula \ref{pvertex}, if $i=1$ and $j=0$ then $i_\ell=j_r=0$ and $j_\ell=i_r=1$ or $i_\ell=j_r=1$ and $j_\ell=i_r=0$. We define: $S(x_{1,0},y_{i',j'})=Minsize_{i',j'\in\{0,1,2\}}\{S_\ell(x_{1,0},y_{i_\ell',j_r'}) \cup s_r(x_{0,1},y_{i_r',j_r'}),S_\ell(x_{0,1},y_{i_\ell',j_r'}) \cup s_r(x_{1,0},y_{i_r',j_r'})\}.$
\begin{figure}[!h]
\centering
\includegraphics[width=6cm]{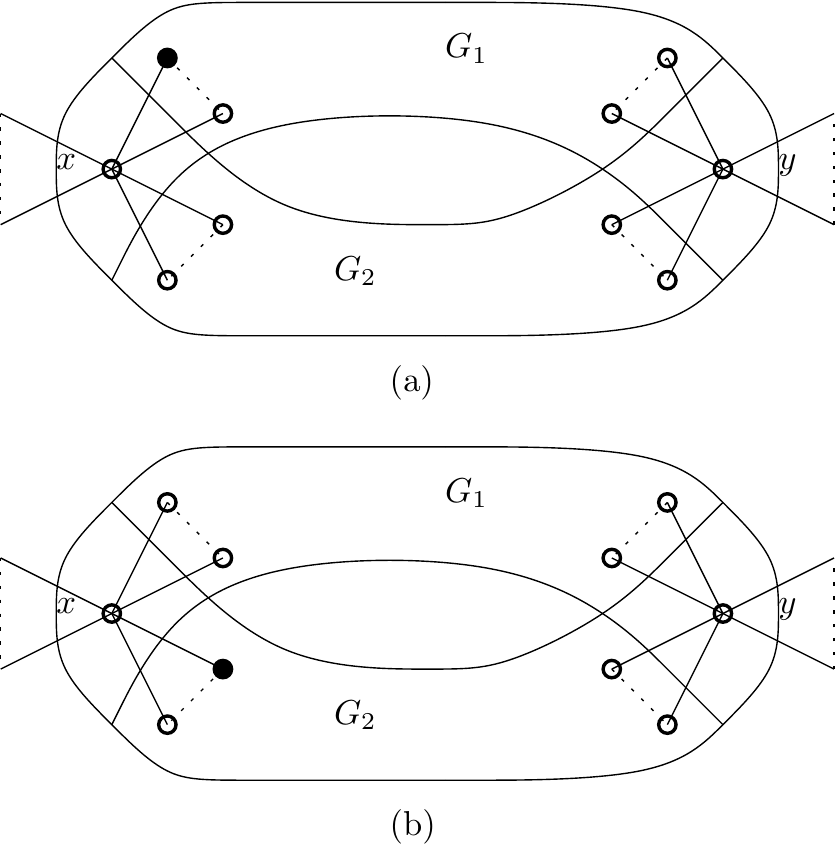}
\caption{ $x$ is  dominated by exactly one vertex of $V(\hat{G})$.}\label{pmethod2}
\end{figure}
\item[$\bullet$]
 $x \notin D$, if  $x$ is  dominated by exactly two vertices of $V(\hat{G})$, then one of the following cases occurs:
  \begin{itemize}
  \item[1.]
  $x$ is  dominated by two vertex of $G_1$ that they are out of $G_1$, see Figure \ref{pmethod3}.a.
  \item[2.]
  $x$ is dominated by  two vertices of $G_2$ that they are out of $G_1$, see Figure \ref{pmethod3}.b.
  \item[3.]
 $x$ is dominated by  exactly one vertex of $G_1$ and exactly one vertex of $G_2$, see Figure \ref{pmethod3}.c.
  
\end{itemize}
 In formula \ref{pvertex}, if $i=2$ and $j=0$ then $(i_\ell,j_\ell(=(0,2)$,  $(2,0)$ or  $(1,1)$  and $i_r=j_\ell,\;j_r=i_\ell$. So that, we define:
 \begin{equation*}
\begin{array}{ll}
 S(x_{2,0},y_{i',j'})=Minsize&\{S_\ell(x_{2,0},y_{i_\ell',j_r'}) \cup s_r(x_{0,2},y_{i'_r,j'_r})\\
&S_\ell(x_{0,2},y_{i_\ell',j_r'}) \cup s_r(x_{2,0},y_{i_r',j_r'}),\\
 &S_\ell(x_{1,1},y_{i_\ell',j_r'}) \cup s_r(x_{1,1},y_{i_r',j_r'})\}.
\end{array}
\end{equation*}

\begin{figure}[!h]
\centering
\includegraphics[width=6cm]{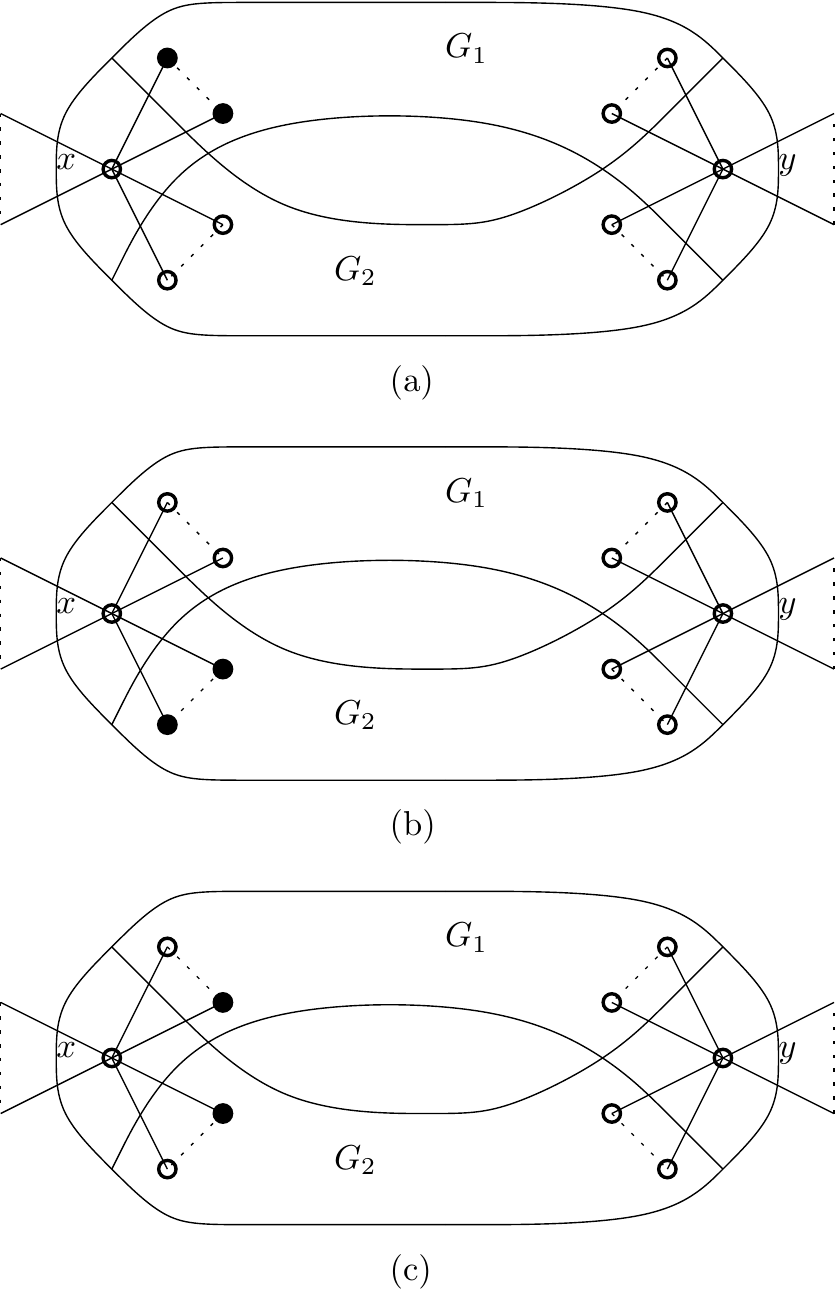}
\caption{$x$ is  dominated by exactly two vertices of $V(\hat{G})$.}\label{pmethod3}
\end{figure}

\item[$\bullet$]
 $x \notin D$, if $x$ is  dominated by exactly one vertex of $\hat{G}$ and exactly one vertex out of $\hat{G}$ then one of the cases hold
 \begin{itemize}
 \item[1.]
  $x$ is  dominated by one vertex of $G_1$ and a vertex out of  $G_1 o_pG_2$ see Figure \ref{pmethod4}.a.
  \item[2.]
   $x$ is  dominated by one vertex of $G_2$ and a vertex out of  $G_1 o_pG_2$ see Figure \ref{pmethod4}.b.
   \end{itemize}

 In formula \ref{pvertex}, if $i=1$ and $j=1$ then $(i_\ell,j_\ell)=(1,1),(i_r,j_r)=(0,2)$ or $(i_\ell,j_\ell)=(0,2),(i_r,j_r)=(1,1)$. So that we define: $S(x_{1,1},y_{i',j'})=Minsize\{S_\ell(x_{1,1},y_{i'_\ell,j'_\ell}) \cup s_r(x_{0,2},y_{i'_r,j'_r}),S_\ell(x_{0,2},y_{i'_\ell,j'_\ell}) \cup s_r(x_{1,1},y_{i'_r,j'_r})\}.$
\begin{figure}[!h]
\centering
\includegraphics[width=6cm]{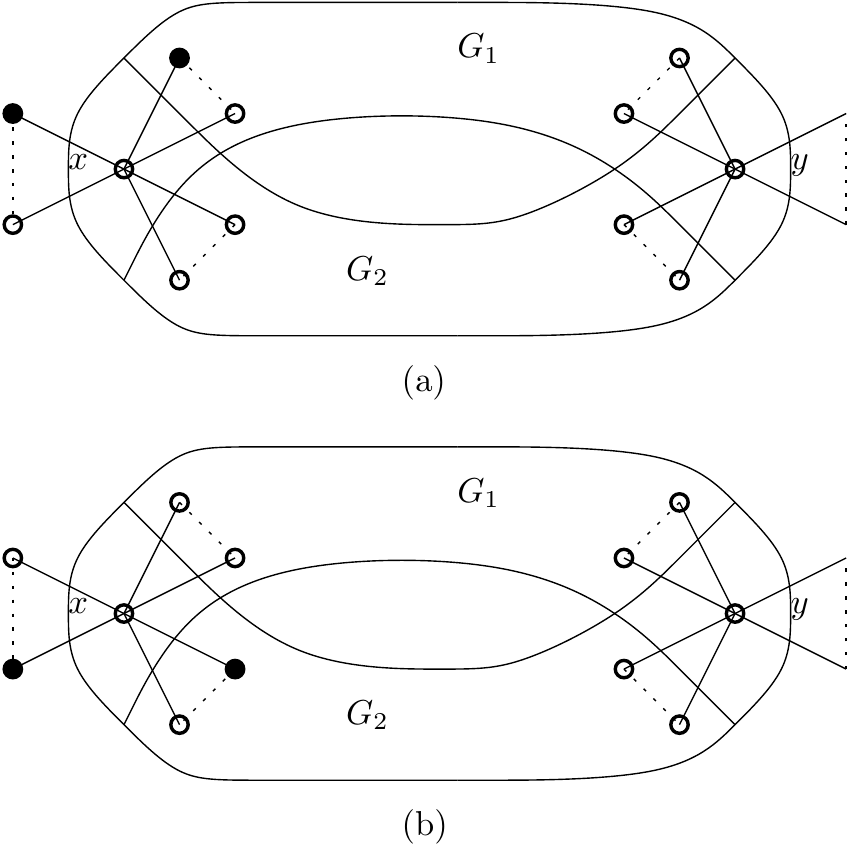}
\caption{$x$ is  dominated by exactly one of $V(\hat{G})$ and one vertex of $V(G) \setminus V(\hat{G})$.}\label{pmethod4}
\end{figure}

 \end{itemize}

\subsection*{Sets for $g$-vertices of $T$}
Let $t$ is a vertex of $T$ is labeled by $(x,y)_g$. In this procedure, the set $S(x_{i,j},y_{i',j'})$ will computed for a given vertex $x$, $y$. The sets corresponding to  $\tau_l(t)$ and $\tau_r(t)$ are $S_\ell(x_{i_\ell,j_\ell},y_{i'_\ell,j'_\ell})$ and $S_r(x_{i_r,j_r},y_{i'_r,j'_r})$ respectively.
\begin{algorithm}[h!]
\begin{algorithmic}[1]
\Procedure{ProcessGvertex}{$x,y,z$}
 \ForAll{$(i,j)\in M$}
                    \State $setlist[i,j]\gets \emptyset$
                      \EndFor

                    \ForAll{$(i,j)\in M$}
                    \ForAll{$(i'_r,j'_r)\in M$}
                    \ForAll{$(i',j')\in \{(0,0),(0,1),(0,2)\}$}
                    \State Add $S_\ell(x_{i,j},y_{i',j'})\cup S_r(y_{i',j'},z_{i'_r,j'_r})$ to $setlist[i',j']$
                     \EndFor
                    \State Add $S_\ell(x_{i,j},y_{0,1})\cup S_r(y_{1,0},z_{i'_r,j'_r})$ to $setlist[1,0]$
                    \State Add $S_\ell(x_{i,j},y_{1,0})\cup S_r(y_{0,1},z_{i'_r,j'_r})$ to $setlist[1,0]$

                    \State Add $S_\ell(x_{i,j},y_{1,1})\cup S_r(y_{0,2},z_{i'_r,j'_r})$ to $setlist[1,1]$
                    \State Add $S_\ell(x_{i,j},y_{0,2})\cup S_r(y_{1,1},z_{i'_r,j'_r})$ to $setlist[1,1]$

                    \State Add $S_\ell(x_{i,j},y_{2,0})\cup S_r(y_{0,2},z_{i'_r,j'_r})$ to $setlist[2,0]$
                    \State Add $S_\ell(x_{i,j},y_{0,2})\cup S_r(y_{2,0},z_{i'_r,j'_r})$ to $setlist[2,0]$
                    \State Add $S_\ell(x_{i,j},y_{1,1})\cup S_r(y_{1,1},z_{i'_r,j'_r})$ to $setlist[2,0]$
                    \EndFor

                    \State $S(x_{i,j},y_{i',j'})\gets Minsize(setlist[i',j'])$
                    \EndFor
\EndProcedure
	\end{algorithmic}
\end{algorithm}
Let the roots of $\tau_l(t)$ and $\tau_r(t)$ are labeled by $(x,y)$ and $(y,z)$ respectively, for some $z \in V$.
Let $s_\ell(x_{i_\ell,j_\ell},y_{i'_\ell,j'_\ell})$ and $s_r(y_{i_r,j_r},z_{i'_r,j'_r})$ be  associated with the vertex $(x,y)$ and $(y,z)$ of $T$.
 If a vertex  $w \in V(\hat{G})$ (resp. $V(G)\setminus V(\hat{G})$) dominate $x$ then, $w\in V(G_1)$ (resp.  $V(G)\setminus V(G_1)$), so $i_\ell=i$ and $j_\ell=j$.  In this kind of operation, number of vertices $V(G_1)$, $(V(G)\setminus V(G_1)$, $V(G_2)$ and $(V(G)\setminus V(G_2)$ make changes in $i'$ and $j'$. So for each $(i,j) \in M$ define:
\begin{equation} \label{gvertex}
S(x_{i,j},y_{i',j'})=Minsize_{(i'_r,j'_r)\in M}\{s_\ell(x_{i,j},y_{i'_\ell,j'_\ell}) \cup s_r(y_{i_r,j_r},z_{i'_r,j'_r})\}
\end{equation}
To find a relation between $(i',j')$, $(i'_\ell,j'_\ell)$ and $(i_r,j_r)$ for different value of $(i'_r,j'_r)$, let $\tau_l(t)$, $\tau_r(t)$ and $\tau(t)$ of $T$ corresponding to  $p$-graphs $G_1$, $G_2$  and $\hat{G}=G_1o_g G_2$.
According to the  number of vertices $D \cap V(G_1)$ and $D\cup V(G_2)$ that dominate $y$, the following cases occur:

\begin{itemize}
 \item[$\bullet$]
$y \in D$, see Figure \ref{gmethod0}.

 In formula \ref{gvertex}, if $i'=0$ and $j'= 0$  then $i'_\ell=j'_\ell=0$, $i_r=j_r=0$ and $j'_r=0$, now we  define:
  $$S(x_{i,j},y_{0,0})=Minsize_{(i'_r,j'_r)\in M}\{s_\ell(x_{i,j},y_{0,0}) \cup s_r(y_{0,0},z_{i'_r,j'_r})\}.$$

\begin{figure}[!h]
\centering
\includegraphics[width=6cm]{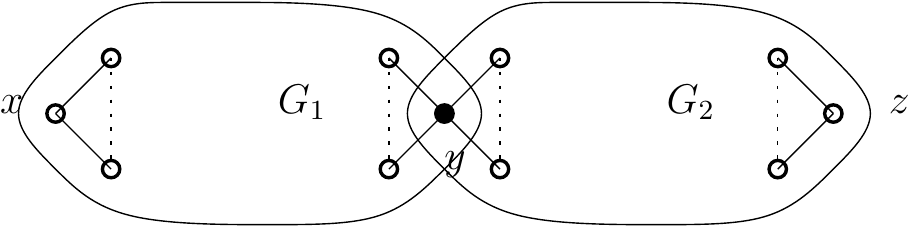}
\caption{$y\in D$}\label{gmethod0}
\end{figure}

\item[$\bullet$]
$y \notin D$, if $y$ is not dominated by any vertex of $\hat{G}$ then, it is dominated by one vertex or  two vertices in $V(G)\setminus V(\hat{G})$, see Figure \ref{gmethod1}.

 In the left side of formula \ref{gvertex}, if $i'=0$ and $j' \neq 0$  then $i'_\ell=i_r=0$ and $j'_\ell=j_r=j$. So that we define
  $$S(x_{i,j},y_{0,1})=Minsize_{(i'_r,j'_r)\in M}\{s_\ell(x_{i,j},y_{0,1}) \cup s_r(y_{0,1},z_{i'_r,j'_r})\},$$
 $$S(x_{i,j},y_{0,2})=Minsize_{(i'_r,j'_r)\in M}\{s_\ell(x_{i,j},y_{0,2}) \cup s_r(y_{0,2},z_{i'_r,j'_r})\}.$$

\begin{figure}[!h]
\centering
\includegraphics[width=6cm]{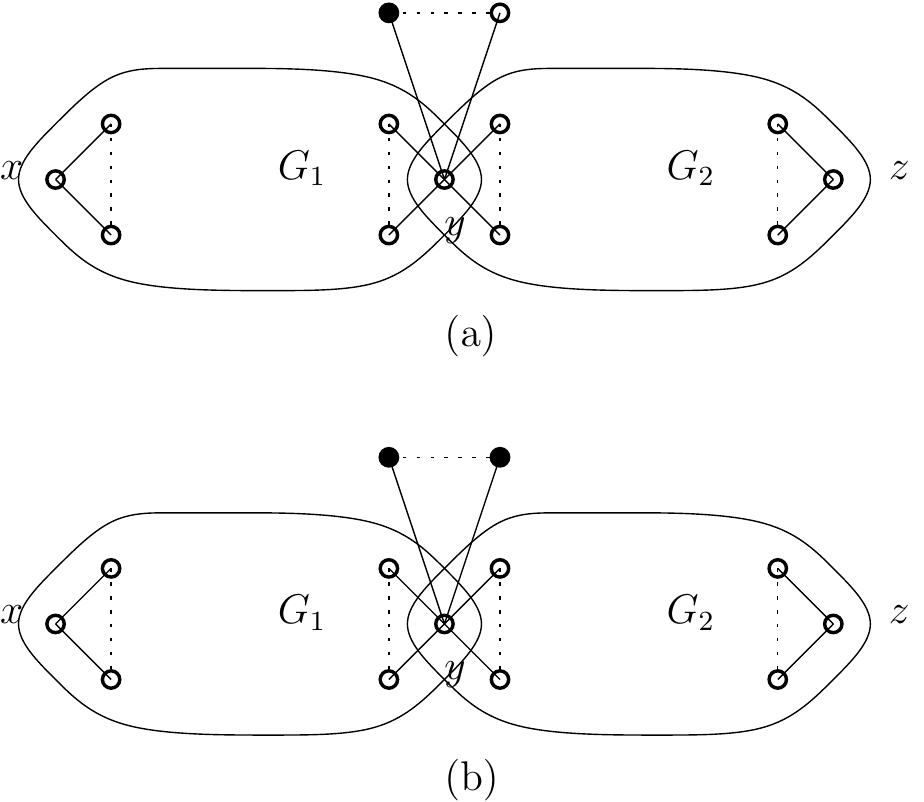}
\caption{$y$ is  dominated by  one vertex or  vertices of $V(G) \setminus V(\hat{G})$.}\label{gmethod1}
\end{figure}

\item[$\bullet$]
$y \notin D$, if $y$ is  dominated by exactly one vertex  $w \in V(\hat{G})$ and it is not dominated by any vertex of $V(G)\setminus V(\hat{G})$, then

 $w \in V(G_1)$, see Figure \ref{gmethod2}.a or $w\in V(G_2)$ see Figure \ref{gmethod2}.b.

 In formula \ref{gvertex}, if $i'=1$ and $j'= 0$  then $(i'_\ell,j'_\ell)=(0,1)\;or\;(1,0)$, $i_r=j'_\ell$ and $j_r=i'_\ell$. So that
 we  define:
  $$S(x_{i,j},y_{1,0})=Minsize_{(i'_r,j'_r)\in M}\{s_\ell(x_{i,j},y_{0,1}) \cup s_r(y_{1,0},z_{i'_r,j'_r}),s_\ell(x_{i_\ell,j_\ell},y_{1,0}) \cup s_r(y_{0,1},z_{i'_r,j'_r})\}.$$
\begin{figure}[!h]
\centering
\includegraphics[width=6cm]{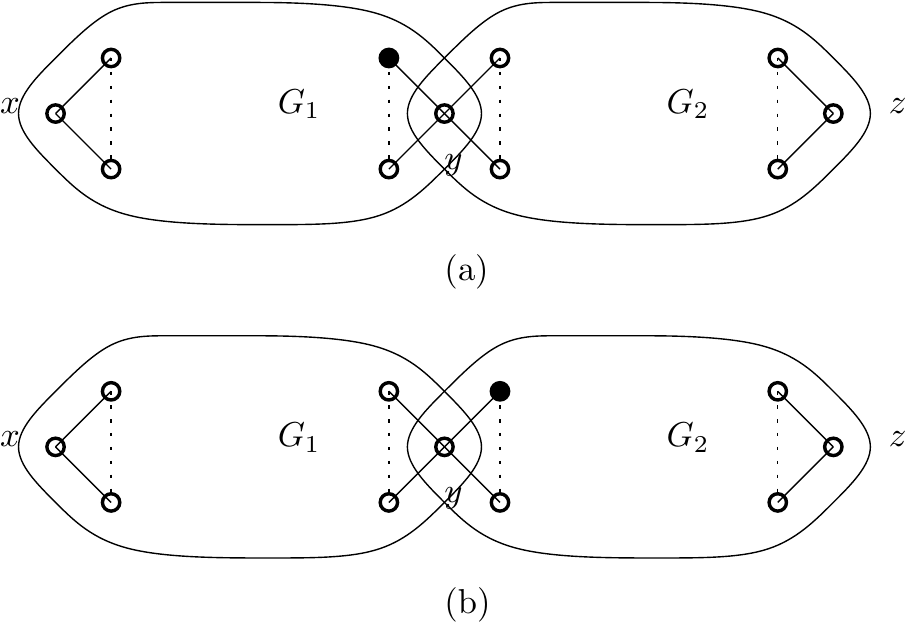}
\caption{$y$ is  dominated by  one vertex  of $ V(\hat{G})$.}\label{gmethod2}
\end{figure}

\item[$\bullet$]
$y \notin D$, let $y$ is  dominated by exactly two vertices of $V(\hat{G})$ and it is not dominated by any $V(G)\setminus V(\hat{G})$
then one of the following cases occur
\begin{itemize}
\item[1.]
$y$ is  dominated by  two vertices of $V(G_1)$, see Figure \ref{gmethod3}.a,
\item[2.]
$y$ is  dominated by  two vertices of $V(G_2)$, see Figure \ref{gmethod3}.b,

\item[3.]
 $y$ is  dominated by one vertex of $V(G_1)$ and one vertex of $V(G_2)$, see Figure \ref{gmethod3}.c.
 \end{itemize}

 In formula \ref{gvertex}, if $i'=2$ and $j'= 0$  then $(i'_\ell,j'_\ell)=(0,2),(2,0)\;or\;(1,1)$ and $i_r=j'_\ell$, $j_r=i'_\ell$.
 we  define:
\begin{equation*}
\begin{array}{ll}
 S(x_{i,j},y_{2,0})=Minsize_{(i'_r,j'_r)\in M}&\{s_\ell(x_{i,j},y_{2,0}) \cup s_r(y_{0,2},z_{i'_r,j'_r}),\\
&s_\ell(x_{i,j},y_{0,2}) \cup s_r(y_{2,0},z_{i'_r,j'_r}),\\
&s_\ell(x_{i,j},y_{1,1}) \cup s_r(y_{1,1},z_{i'_r,j'_r})\}.
\end{array}
\end{equation*}

\begin{figure}[!h]
\centering
\includegraphics[width=6cm]{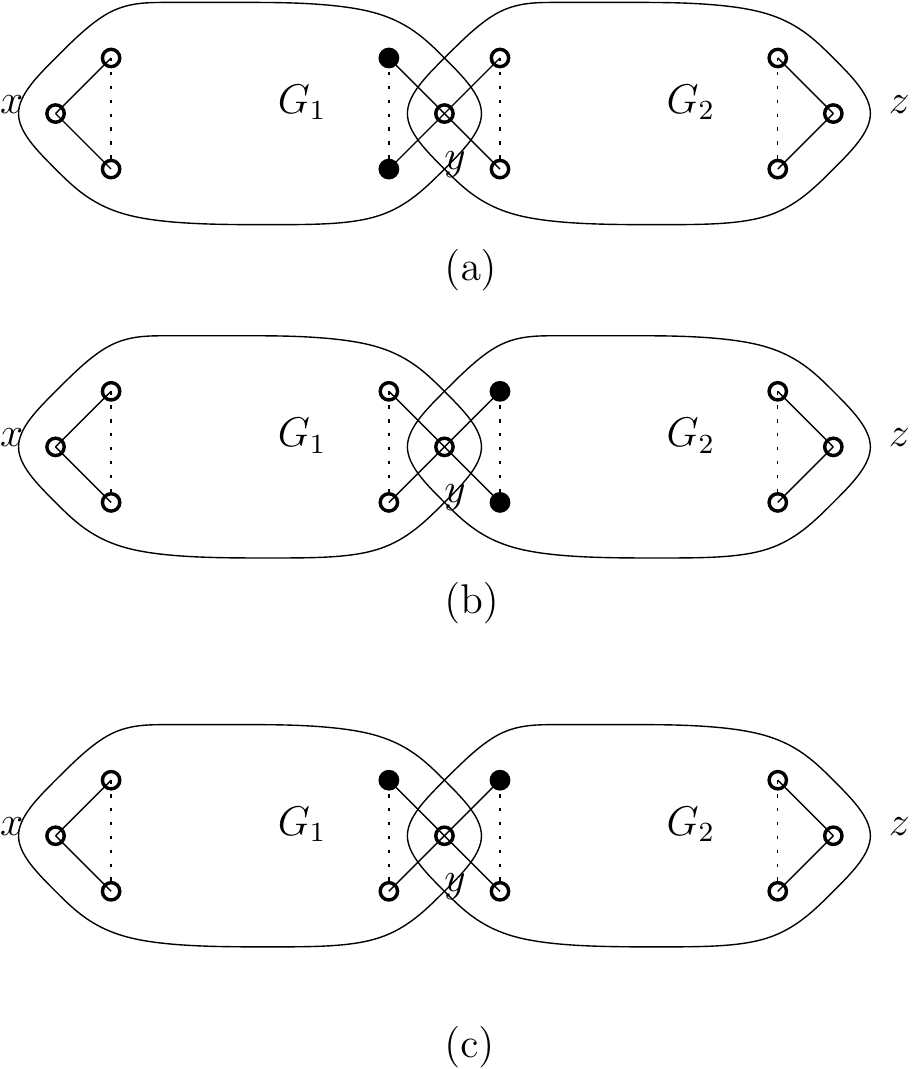}
\caption{$y$ is  dominated by  two vertices of $V(\hat{G})$.}\label{gmethod3}
\end{figure}

\item[$\bullet$]
$y \notin D$, if $y$ is  dominated by one vertex of $\hat{G}$ and one vertex of $V(G)\setminus V(\hat{G})$, then
$y$ is dominated by one vertex of $V(G_1)$ and the other out of $V(G_1o_gG_2)$ see Figure \ref{gmethod4}.a or one vertex of $V(G_2)$ and the other out of $V(G_1o_gG_2)$ see Figure \ref{gmethod4}.b.

 In formula \ref{gvertex}, if $i'=1$ and $j'= 1$  then $(i'_\ell,j'_\ell)=(1,1)$ and $(i_r,j_r)=(0,2)$ or$(i'_\ell,j'_\ell)=(0,2)$ and $(i_r,j_r)=(1,1)$.
 So that we  define:
 $$S(x_{i,j},y_{1,1})=Minsize_{i'_r,j'_r\in\{0,1,2\}}\{s_\ell(x_{i,j},y_{1,1}) \cup s_r(y_{0,2},z_{i'_r,j'_r}),s_\ell(x_{i_\ell,j_\ell},y_{0,2}) \cup s_r(y_{1,1},z_{i'_r,j'_r})\}.$$

\begin{figure}[!h]
\centering
\includegraphics[width=6cm]{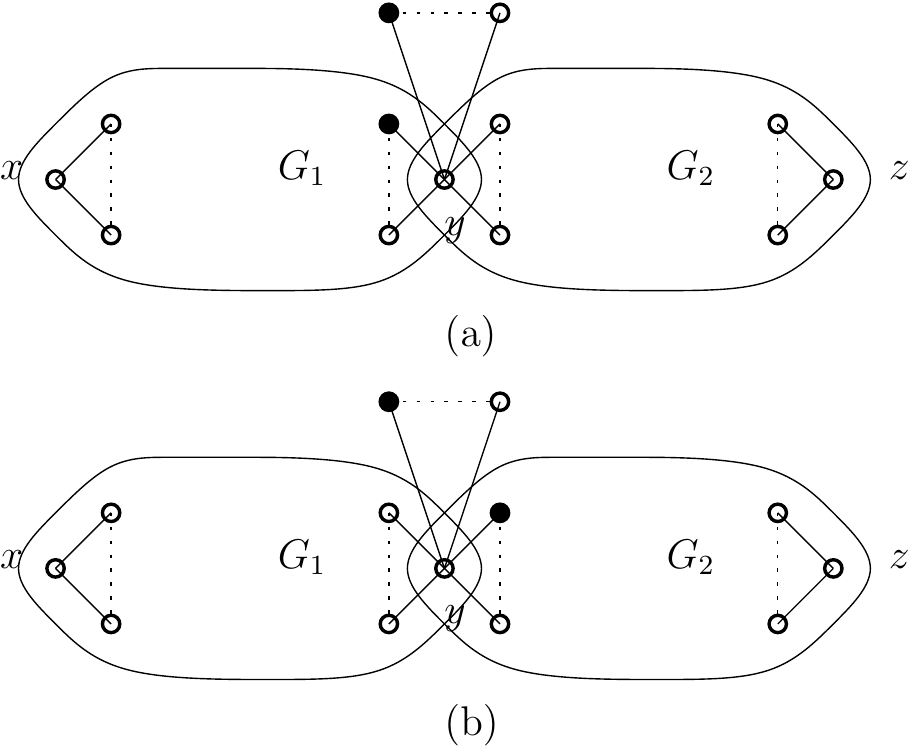}
\caption{$y$ is  dominated by  one vertex of $V(\hat{G})$ and one  vertex of $V(G) \setminus V(\hat{G})$.}\label{gmethod4}
\end{figure}
\end{itemize}
%\end{itemize}
In each formula of subprocedures  $s$-vertex, $p$-vertex and $g$-vertex, we delete any undefinable subformula of the right side of formula computing $S_i(x_{i,j},y_{i',j'})$. Furthermore if all subformulas are deleted from the right hand side of some formula, then $S_i(x_{i,j},y_{i',j'})$ will be undefinable.

\section{Algorithm $[1,2]$-MINSET}
At the first step of algorithm  $[1,2]$-MINSET finding a parse tree for given GSP like $T$ by some known linear algorithms. Each vertex $v$ of $T$ have a label such that $(x,y)_i$, $(x,y)_s$, $(x,y)_p$ or $(x,y)_g$, such that $x$ and $y$ are terminals of corresponding $p$-graph to $\tau(v)$. $x$ and $y$ will be  inputs of procedures. By  traversing parse tree $T$ in bottom-up order and visiting vertices of $T$, proper procedure will be called. After visiting root of $T$ and computing  $S(x_{i,j},y_{i',j'})$ for root$(T)$, a $[1,2]$-set for $G$ can be found. It is enough to return a $S(x_{i,j},y_{i',j'})$ with minimum cardinality when {$(i,j),(i',j')\in \{(0,0),(0,1),(0,2)\}$}.

\algnewcommand\algorithmicswitch{\textbf{switch}}
\algnewcommand\algorithmiccase{\textbf{case}}
\algnewcommand\algorithmicassert{\texttt{assert}}
\algnewcommand\Assert[1]{\State \algorithmicassert(#1)}%
\algdef{SE}[SWITCH]{Switch}{EndSwitch}[1]{\algorithmicswitch\ #1\ \algorithmicdo}{\algorithmicend\ \algorithmicswitch}%
\algdef{SE}[CASE]{Case}{EndCase}[1]{\algorithmiccase\ #1}{\algorithmicend\ \algorithmiccase}%
\algtext*{EndSwitch}%
\algtext*{EndCase}%

\begin{algorithm}[h]
	\caption{:Finding a $\gamma_{[1,2]}$-sets of a GSP}
\begin{algorithmic}[1]

        \State find a parse tree of $G$ like $T$
        \State $S \gets  empty\; stack$ \Comment{ main stack with nodes from X}
        \State $Q  \gets empty\; queue$ \Comment{auxiliary queue to build S}
        \State $Q.add(root\; r)$
        \State $M \gets \{(0,0),(0,1),(0,2),(1,0),(1,1),(2,0)\}$
        \While{$Q \neq \emptyset$}
        \State $x  \gets Q.pop()$
        \State $Q.add(ch(x))$
        \State $S.add(x)$
        \EndWhile

        \While{$S \neq \emptyset$}
        \State $v  \gets S.pop()$ \Comment{all children of v are processed}
        %\State $Q.add(ch(x))$
        %\State $S.add(x)$

  \Switch{$type\;of\;v$}
    \Case{$Leaf$}
    \State ProcessLeaf$(x,y)$
    \Comment{$(x,y)_i$ is the label of $v$}
    \EndCase
    \Case{$s-vertex$}
      \State ProcessSvertex$(x,z,y)$
      \Comment{$(x,y)_s$, $(x,z)$ and $(z,y)$ are label of $v$, left and right child of $v$ respectively. }
    \EndCase

    \Case{$p-vertex$}
     \State ProcessPvertex$(v)$
     \Comment{$(x,y)_p$ is the label of $v$ and the label of left and right child of $v$  is $(x,y)$. }
                     \EndCase

    \Case{$g-vertex$}
    \State ProcessGvertex$(x,y,z)$
    \Comment{$(x,y)_s$, $(x,y)$ and $(y,z)$ are label of $v$, left and right child of $v$ respectively.}
    \EndCase

  \EndSwitch
  \EndWhile
		
\State $D\gets \emptyset$
          \ForAll{$(i,j),(i',j')\in \{(0,0),(0,1),(0,2)\}$}
        \State Add $S(x_{i,j},y_{i',j'})$ to $D$
         \EndFor
          \State $\gamma_{[1,2]}(G)\gets Minsize(D)$
\end{algorithmic}
\end{algorithm}
\begin{theorem}\label{Complexity}
	For a given generalized series-parallel graph $G=(V,E)$, algorithm $[1,2]$-MINSET find a minimum $[1,2]$-set in time $O(\vert V \vert)$.
	\end{theorem}	

\begin{proof}
Algorithm  $[1,2]$-MINSET, traversing parse tree $T$ in bottom-up order, easily computes at most $36$ set for each internal vertex of them.
 Each initial sets for leaves of tree represents all possible $[1,2]$-dominating sets in a graph containing only one edge. Let
$G_1 = (V_1, E_1)$ and $G_2=(V_2, E_2)$ be the graphs represented by the subtrees $\tau_l(t)$ and $\tau_r(t)$. Assume these are input of procedures ProcessLeaf, ProcessSvertex,   ProcessPvertex and ProcessGvertex. It is easy to see that this procedure  finds all possible minimum $[1,2]$-sets in graph. Finally, extracts only a valid minimum $[1,2]$-set. These step of algorithm require at most $\vert O(V_T) \vert$ operations. Since each binary tree with $n$ leaves has $O(n)$ vertices and the binary parse tree of every GSP graph have has $\vert E(G) \vert$ leaves, so $\vert V_T\vert \in O(\vert E (G)\vert)$. Since every GSP graphs are planar and in planar graph we have $\vert E \vert \leq 3 \vert V \vert - 6$. So in order to prove that algorithm  $[1,2]$-MINSET computes a $\gamma_{[1,2]}$-set for a given GSP graph in time $O(\vert V \vert)$, it is sufficient to show that the parse tree $T$ can be constructed in $O(\vert V \vert)$ operations that it is proved in \cite{kikuno1983linear}.
\end{proof}

\section{Algorithms for minimum total $[1,2]$-set}\label{algorithm2}
In this section, we make some changes in three procedures of $[1,2]-MINSET$ algorithm and define  new procedures  ProcessLeaf-t, ProcessSvertex-t,   ProcessPvertex-t and ProcessGvertex-t to use in $total\;[1,2]-MINSET$ algorithm and compute a $\gamma_{t[1,2]}$-set for a GSP graph $G$ while binary parse tree of $G$ is computed in the first step  of algorithm.
These procedures compute set $S(x_{i,j,k},y_{i',j',k'})$ for leaves labeled by $(x,y)_i$ and internal vertices labeled by $(x,y)_s$, $(x,y)_p$, $(x,y)_g$ of in binary parse tree for all $i,j,i',j' \in \{0,1,2\}$ and $k,k' \in \{0,1\}$.
% $S(x_{i,j,k},y_{i',j',k'})$ computes  total $[1,2]$-dominating set with minimum size for the $p$-graph $\hat{G}$.

There is a connection between the sets computed for each vertex of $T$ and the $p$-graph $\hat{G}=(\hat{V},\hat{E})$. Let $D$ is a total $[1,2]$-dominating set for $\hat{G}$, then for each vertex $x\in V(\hat{G})$, $x_{i,j,k}$ implies the following assertions with respect to $i$, $j$, $k$ and $D$.

\begin{itemize}
\item[$\bullet$]
$x_{i,j,0}$ implies that $x \notin D$, $i$ vertices of $V(\hat{G})$ and $j$ vertices of $V(G)\setminus V(\hat{G})$ that they adjacent to $x$ belong to $D$.
\item[$\bullet$]
$x_{i,j,1}$ implies that $x \in D$, $i$ vertices of $V(\hat{G})$ and $j$ vertices of $V(G)\setminus V(\hat{G})$ that they adjacent to $x$ belong to $D$.
\end{itemize}

Since every vertices of graph must be dominated by one vertex or two vertices of total $[1,2]$-set, it is clear,
for all vertices of tree $S(x_{i,j,k},y_{i',j',k'})$ is undefinable, when $(i,j),(i',j')\in\{(0,0),(1,2),(2,1),(2,2)\}$ and we define a new set as follow,
$$M_t=\{(0,1,0),(0,1,1),(1,0,0),(1,0,1),(0,2,0),(0,2,1),(2,0,0),(2,0,1),(1,1,0),(1,1,1)\}.$$
At the end of this algorithm   a $\gamma_{t[1,2]}$-set for a GSP graph $G$ with minimum size can be found.
\subsection*{Procedure for Leaves of $T$}
Input of this procedure is a leaf of $v\in V_T$  labeled by $(x,y)_i$ and output $S(x_{i,j,k},y_{i',j',k'})$ for the leaf labeled by $(x,y)_i$ and for all $(i,j,k),(i',j' ,k')\in M_t$.
We compute $S(x_{i,j,k},y_{i',j',k'})$ for leaves using  table \ref{table1}.

In this procedure, following cases will occur:
\begin{itemize}
  \item
$k=k'=1$:  $x,y\in D$ and  $S(x_{i,j,k},y_{i',j',k'})=\{x,y\}$.
 \item
 $i=1$, $k'=1$: $x$ is dominated by $y$ and  $S(x_{i,j,k},y_{i',j',k'})=\{y\}$.
 \item
$i'=1$, $k=1$: $y$ is dominated by $x$ and  $S(x_{i,j,k},y_{i',j',k'})=\{x\}$.
\item
$k=k'=0$: $x$ and $y$ are dominated by one vertex or two vertices out of this leaf, so $x,y\notin D$ and  $S(x_{i,j,0},y_{i',j',0})=\emptyset$.
\item
$i=0$ and $k'=1$: $S(x_{i,j,k},y_{i',j',k'})$ is undefinable. $i=0$ means that there is not any vertex of $\hat{G}=(\hat{V},\hat{E})$ to dominate $x$, so $y \notin D$ and $S(x_{0,j,k},y_{i',j',k'})$ is definable if $k'=0$. Similarly if $i'=0$ and $k=1$ then,  $S(x_{i,j,k},y_{i',j',k'})$ is undefinable.
\item
$i=1$ and $k'=0$:  $S(x_{i,j,k},y_{i',j',k'})$ is undefinable. $i=1$ means that there is a vertex in $\hat{G}$ to dominate $x$, so it is necessary to  set $y\in D$ and  nd $S(x_{1,j,k},y_{i',j',k'})$ is definable if $k'=1$ . Similarly if $i'=1$ and $k=0$ then  $S(x_{i,j,k},y_{i',j',k'})$ is undefinable.
\item
$i=2$: since there is exactly one vertex $y\in V(\hat{G})$ to dominate $x$, $S(x_{i,j,k},y_{i',j',k'})$ is undefinable and similarly for $i'=2$,  $S(x_{i,j,k},y_{i',j',k'})$ is undefinable.

\end{itemize}

\subsection*{Sets for $s$-vertices of $T$}
Let $v$ is a vertex of $T$ is labeled by $(x,y)_s$. In this procedure, the set $S(x_{i,j,k},y_{i',j',k'})$ will computed for a given vertex $x$, $y$ and common vertex $z$. Assume the sets corresponding to  $\tau_l(t)$ and $\tau_r(t)$ are $S_\ell(x_{i_\ell,j_\ell,k_\ell},z_{i'_\ell,j'_\ell,k'_\ell})$ and $S_r(z_{i_r,j_r,k_r},y_{i'_r,j'_r,k'_r})$ respectively.

Now for each $(i,j,k) \in M_t\}$ define:
\begin{equation} \label{svertext}
S(x_{i,j,k},y_{i',j',k'})=Minsize_{(i'_\ell,j'_\ell,k'_\ell)\in M}\{s_\ell(x_{i,j,k},z_{i'_\ell,j'_\ell,k'_\ell}) \cup s_r(z_{j'_\ell,i'_\ell,k'_\ell},y_{i',j',k'})\}
\end{equation}

$G_1=(V_1,E_1,x,z)$, $G_2=(V_2,E_2,z,y)$ and  $\hat{G}=G_1 o_sG_2=(\hat{V},\hat{E},x,y)$, to prove formula \ref{svertext}, we describe following cases.
\begin{itemize}
\item[$\bullet$]
$z\in D$, if $z$ is dominated by two vertices of $V(G_1)\setminus V(G_2)$, $V(G_2)\setminus V(G_1)$ or one of them of $V(G_1)$ and the other of $V(G_2)$  then, $(2,0,1), (0,2,1), (1,1,1) \in M_t$.
\item[$\bullet$]
$z\notin D$, if $z$ is dominated by exactly one vertex of $V(G_1)\setminus V(G_2)$ then, $(1,0,0)\in M$ and if it   is dominated by exactly one vertex of $V(G_2)\setminus V(G_1)$ then, $(0,1,0)\in M_t$.
\item[$\bullet$]
$z\in D$, if $z$ is dominated by exactly one vertex of $V(G_1)\setminus V(G_2)$ then, $(1,0,1)\in M$ and if it   is dominated by exactly one vertex of $V(G_2)\setminus V(G_1)$ then, $(0,1,1)\in M_t$.
 \item[$\bullet$]
$z\notin D$, if $z$ is dominated by two vertices of $V(G_1)\setminus V(G_2)$, $V(G_2)\setminus V(G_1)$ or one of them of $V(G_1)$ and the other of $V(G_2)$  then, $(2,0,0), (0,2,0), (1,1,0) \in M_t$.

\end{itemize}

\subsection*{Sets for $p$-vertices of $T$}
Let $v$ is a vertex of $T$ is labeled by $(x,y)_p$, in this procedure, the set $S(x_{i,j},y_{i',j'})$ will computed for a given vertex $x$, $y$. The sets corresponding to  $\tau_l(t)$ and $\tau_r(t)$ are $S_\ell(x_{i_\ell,j_\ell},y_{i'_\ell,j'_\ell})$ and $S_r(x_{i_r,j_r},y_{i'_r,j'_r})$ respectively.
%
%Let vertex $t$ of parse tree is labeled by $(x,y)_p$, roots $\tau_l(t)$  and $\tau_r(t)$ are $(x,y)$.
%$s_\ell(x_{i_\ell,j_\ell,k_\ell},y_{i'_\ell,j'_\ell,k'_\ell})$ and $s_r(x_{i_r,j_r,k_r},y_{i'_r,j'_r,k'_r})$ be sets associated with the $\tau_l(t)$  and $\tau_r(t)$.
 For $(i,j,k)\in M_t$, we define,
\begin{equation} \label{pvertext}
S(x_{i,j,k},y_{i',j',k'})=Minsize\{s_\ell(x_{i_\ell,j_\ell,k},y_{i'_\ell,j'_\ell,k'})\cup s_r(x_{i_r,j_r,k},y_{i'_r,j'_r,k'})\}.
\end{equation}

 Let $\tau_l(t)$, $\tau_r(t)$ and $\tau(t)$ represent $p$-graphs $G_1=(V_1,E_1,x,y)$, $G_2=(V_2,E_2,x,y)$ and  $\hat{G}=G_1 o_pG_2=(\hat{V},\hat{E},x,y)$ respectively.
According to values of $(i,j,k)$ (resp. $(i',j',k')$) proper values for $(i_\ell,j_\ell,k_\ell)$ (resp. $(i'_\ell,j'_\ell,k'_\ell)$)  and $(i_r,j_r,k_r)$ (resp. $(i'_r,j'_r,k'_r)$) will determine. Moreover, if $x\in D$ then in formula \ref{pvertext} $k=k_\ell=k_r=1$ and if $x\notin D$, then $k=k_\ell=k_r=0$.
 \begin{itemize}
\item[$\bullet$]
If $x$ is  dominated by one vertex of $V(G)\setminus V(\hat{G})$ then, in formula \ref{pvertext}, $i=0,j=1$ and we have $i_\ell=i_r=0$, $j_\ell=j_r=1$. So that we define,
We define:
\begin{equation*}
S(x_{0,1,k},y_{i',j',k'})=Minsize\{s_\ell(x_{0,1,k},y_{i'_\ell,j'_\ell,k'_\ell})\cup s_r(x_{0,1,k},y_{i'_,j'_r,k'_r})\}
\end{equation*}
\item[$\bullet$]
If $x$ is  dominated by two vertices of $V(G)\setminus V(\hat{G})$ then, in formula \ref{pvertext}, $i=0,j=2$ and we have $i_\ell=i_r=0$, $j_\ell=j_r=2$. So that we define,
We define:
\begin{equation*}
S(x_{0,2,k},y_{i',j',k'})=Minsize\{s_\ell(x_{0,2,k},y_{i'_\ell,j'_\ell,k'_\ell})\cup s_r(x_{0,2,k},y_{i'_,j'_r,k'_r})\}
\end{equation*}
\item[$\bullet$]
If $x$ is  dominated by one vertex of $w\in V(\hat{G})$  then, in formula \ref{pvertext}, $i=1$ and $j=0$. Two following cases occur,
\begin{itemize}
  \item [1.]
  $w\in V(G_1)\setminus V(G_2)$, so that $i_\ell=j_r=1$ and $i_r=j_\ell=0$,
  \item [2.]
 $w\in V(G_2)\setminus V(G_1)$, so that $i_\ell=j_r=0$ and $i_r=j_\ell=1$.
\end{itemize}
We define:
\begin{equation*}
\begin{array}{ll}
S(x_{1,0,k},y_{i',j',k'})=Minsize&\{s_\ell(x_{1,0,k},y_{i'_\ell,j'_\ell,k'_\ell})\cup s_r(x_{0,1,k},y_{i'_r,j'_r,k'_r}),\\
&s_\ell(x_{0,1,k},y_{i'_\ell,j'_\ell,k'_\ell})\cup s_r(x_{1,0,k},y_{i'_r,j'_r,k'_r})\}.
\end{array}
\end{equation*}
\item[$\bullet$]
If $x$ is  dominated by two vertices of $v,w\in V(\hat{G})$  then, in formula \ref{pvertext}, $i=2$ and $j=0$. Following cases occur,
\begin{itemize}
  \item [1.]
  $v,w\in V(G_1)\setminus V(G_2)$, so that $i_\ell=j_r=2$ and $i_r=j_\ell=0$,
  \item [2.]
 $v,w\in V(G_2)\setminus V(G_1)$, so that $i_\ell=j_r=0$ and $i_r=j_\ell=2$,
 \item [2.]
 $w\in V(G_1)\setminus V(G_2)$ and $v\in V(G_2)\setminus V(G_1)$, so that $i_\ell=j_\ell=i_r=j_r=1$.
\end{itemize}
We define:
\begin{equation*}
\begin{array}{ll}
S(x_{2,0,k},y_{i',j',k'})=Minsize&\{s_\ell(x_{2,0,k},y_{i'_\ell,j'_\ell,k'_\ell})\cup s_r(x_{0,2,k},y_{i'_r,j'_r,k'_r}),\\
&s_\ell(x_{0,2,k},y_{i'_\ell,j'_\ell,k'_\ell})\cup s_r(x_{2,0,k},y_{i'_r,j'_r,k'_r}),\\
&s_\ell(x_{1,1,k},y_{i'_\ell,j'_\ell,k'_\ell})\cup s_r(x_{1,1,k},y_{i'_r,j'_r,k'_r})\}.
\end{array}
\end{equation*}
\item[$\bullet$]
If $x$ is  dominated by one vertex of $v\in V(\hat{G})$  and  one vertex of $w\in V(G)\setminus V(\hat{G})$ then, in formula \ref{pvertext}, $i=1$ and $j=1$. Following cases occur,
\begin{itemize}
  \item [1.]
  $v\in V(G_1)$, so that $i_\ell=j_\ell=1$, $i_r=0$ and $j_r=2$,
  \item [2.]
 $v\in V(G_2))$, so that $i_\ell=0$, $j_\ell=2$ and $i_r=j_r=1$.
\end{itemize}
\begin{equation*}
\begin{array}{ll}
S(x_{1,1,k},y_{i',j',k'})=Minsize&\{s_\ell(x_{1,1,k},y_{i'_\ell,j'_\ell,k'_\ell})\cup s_r(x_{0,2,0},y_{i'_r,j'_r,k'_r}),\\
&s_\ell(x_{0,2,0},y_{i'_\ell,j'_\ell,k'_\ell})\cup s_r(x_{1,1,0},y_{i'_r,j'_r,k'_r})\}.
\end{array}
\end{equation*}

\end{itemize}

\subsection*{Sets for $g$-vertices of $T$}
Let $v$ is a vertex of $T$ is labeled by $(x,y)_g$. In this procedure, the set $S(x_{i,j},y_{i',j'})$ will computed for a given vertex $x$, $y$.

Let the roots of $\tau_l(t)$ and $\tau_r(t)$ are labeled by $(x,y)$ and $(y,z)$ respectively, for some $z \in V$.
 The sets corresponding to  $\tau_l(t)$ and $\tau_r(t)$ are $s_\ell(x_{i_\ell,j_\ell,k_\ell},y_{i'_\ell,j'_\ell,k'_\ell})$ and $S_r(x_{i_r,j_r,k_r},y_{i'_r,j'_r,k'_r})$ respectively.
 If a vertex  $w \in V(\hat{G})$ (resp. $V(G)\setminus V(\hat{G})$) dominate $x$ then, $w\in V(G_1)$ (resp.  $V(G)\setminus V(G_1)$), so $i_\ell=i$ and $j_\ell=j$.  In this kind of operation, number of vertices $V(G_1)$, $(V(G)\setminus V(G_1)$, $V(G_2)$ and $(V(G)\setminus V(G_2)$ make changes in $i',j'$and $k'$. So for each $(i,j,k) \in M_t$ define:
\begin{equation} \label{gvertext}
S(x_{i,j,k},y_{i',j',k'})=Minsize\{s_\ell(x_{i,j,k},y_{i'_\ell,j'_\ell,k'}) \cup s_r(y_{i_r,j_r,k'},z_{i'_r,j'_r,k'_r})\}
\end{equation}
To find a relation between $(i',j',k')$, $(i'_\ell,j'_\ell,k'_\ell)$ and $(i_r,j_r,k_r)$ for different value of $(i'_r,j'_r,k'_r)$,
let $\tau_l(t)$, $\tau_r(t)$ and $\tau(t)$ of $T$ corresponding to  $p$-graphs $G_1=(V_1,E_1,x,y)$, $G_2=(V_1,E_1,y,z)$  and $\hat{G}=G_1o_g G_2=(\hat{V},\hat{E},x,y)$.
In formula \ref{gvertext}, if $y\in D$, then $k'=k_\ell=k_r=1$, if $y\in D$, then $'=k_\ell=k_r=0$.
According to the  number of vertices $D \cap V(G_1)$ and $D\cup V(G_2)$ that dominate $y$, the following cases occur:

\begin{itemize}
\item[$\bullet$]
$y$ is dominated by only one vertex of $V(G)\setminus V(\hat{G})$,
in the left side of formula \ref{gvertext}, $i=0$ and $j=1$ , so  $i'_\ell=i_r=0$ and $j'_\ell=j_r=1$. We  define:
 $$S(x_{i,j,k},y_{0,1,k'})=Minsize\{s_\ell(x_{i,j,k},y_{0,1,k'}) \cup s_r(y_{0,1,k'},z_{i'_r,j'_r,k'_r})\}$$

\item[$\bullet$]
$y$ is dominated by only two vertices of $V(G)\setminus V(\hat{G})$  then,
in the left side of formula \ref{gvertext}, $i=0$ and $j=2$ , so  $i'_\ell=i_r=0$ and $j'_\ell=j_r=2$. We  define:
 $$S(x_{i,j,k},y_{0,2,k'})=Minsize\{s_\ell(x_{i,j,k},y_{0,2,k'}) \cup s_r(y_{0,2,k'},z_{i'_r,j'_r,k'_r})\}$$

\item[$\bullet$]
$y$ is  dominated by exactly one vertex of $V(\hat{G})$, in formula \ref{gvertext}, if $i=1$ and $j= 0$ . So two cases occur
\begin{itemize}
\item[1.]
$i'_\ell=j_r=0$ and $i_r=j'_\ell=1$,
\item[2.]
 $i'_\ell=j_r=1$ and $i_r=j'_\ell=0$.
 \end{itemize}
 We  define:
 \begin{equation*}
\begin{array}{ll}
S(x_{i,j,k},y_{1,0',k'})=Minsize&\{s_\ell(x_{i,j,k},y_{1,0,k'}) \cup s_r(y_{0,1,k'},z_{i'_r,j'_r,k'_r}),\\
&s_\ell(x_{i_\ell,j_\ell,k_\ell},y_{1,0,k'}) \cup s_r(z_{0,1,k'},y_{i'_r,j'_r,k'_r})\}.
\end{array}
\end{equation*}
\item[$\bullet$]
 $y$ is  dominated by only  two vertices of $V(\hat{G})$, formula \ref{gvertext}, $i=2$ and $j= 0$  so, one of the following cases occur
 \begin{itemize}
 \item[1.]
 $i'_\ell=j_r=0$ and $i_r=j'_\ell=2$,
 \item[2.]
 $i'_\ell=j_r=2$ and $i_r=j'_\ell=0$,
 \item[3.]
 $i'_\ell=j_r=1$ and $i_r=j'_\ell=1$.
 \end{itemize}
 We  define:
\begin{equation*}
\begin{array}{ll}
 S(x_{i,j,k},y_{2,0,k'})=Minsize&\{s_\ell(x_{i,j,k},y_{2,0,k'}) \cup s_r(y_{0,2,k'},z_{i'_r,j'_r,k'_r}),\\
&s_\ell(x_{i,j,k},y_{0,2,k'}) \cup s_r(y_{2,0,k'},z_{i'_r,j'_r,k'_r}),\\
&s_\ell(x_{i,j,k},y_{1,1,k'}) \cup s_r(y_{1,1,k'},z_{i'_r,j'_r,k'_r})\}.
\end{array}
\end{equation*}
\item[$\bullet$]
$y$ is dominated by exactly one vertex of $V(\hat{G})$ and one vertex of $V(G)\setminus V(\hat{G})$, in formula \ref{gvertext} $i=1$ and $j= 1$  so, two cases occur
\begin{itemize}
\item[1.]
$i'_\ell=j'_\ell=1$, $i_r=0$ and $j'_r=1$ ,
\item [2.]
$i'_\ell=0$, $j'_\ell=1$ and $i_r=j_r=1$.
\end{itemize}
 We  define:
 \begin{equation*}
\begin{array}{ll}
S(x_{i,j,k},y_{1,1,k'})=Minsize&\{s_\ell(x_{i,j,k},y_{1,1,k'}) \cup s_r(y_{0,2,k'},z_{i'_r,j'_r,k'_r}),\\
&s_\ell(x_{i,j,k},y_{0,2,k'}) \cup s_r(y_{1,1,k'},z_{i'_r,j'_r,k'_r})\}.
\end{array}
\end{equation*}
\end{itemize}

Similar to procedure internal-set of algorithm $[1,2]$-MINSET,  in each formula of methods  $s$-vertex, $p$-vertex and $g$-vertex, we delete any undefinable subformula of the right side of formula computing $S_i(x_{i,j,k},y_{i',j',k'})$.
 Furthermore if all subformulas are deleted from the right hand side of some formula, then $S(x_{i,j,k},y_{i',j',k'})$ will be undefinable.

Algorithm total $[1,2]$-MINSET is similar to algorithm $[1,2]$-MINSET when we use procedures  ProcessLeaf-t, ProcessSvertex-t,   ProcessPvertex-t and ProcessGvertex-t.

\begin{algorithm}[h]
	\caption{:Finding a $\gamma_{t[1,2]}$-sets of a GSP}
\begin{algorithmic}[1]

        \State find a parse tree of $G$ like $T$
        \State $M_t \gets \{(0,1,0),(0,1,1),(1,0,0),(1,0,1),(0,2,0),(0,2,1),(2,0,0),(2,0,1),(1,1,0),(1,1,1)\}.$
        \State traversing $T$ in a bottom up order and visit vertex $v$
        %\State $Q.add(ch(x))$
        %\State $S.add(x)$

        \Switch{$type\;of\;v$}
            \Case{$Leaf$}
                \State ProcessLeaf$(x,y)$
                \Comment{$(x,y)_i$ is the label of $v$}
            \EndCase
            \Case{$s-vertex$}
                \State ProcessSvertex$(x,z,y)$
                \Comment{$(x,y)_s$, $(x,z)$ and $(z,y)$ are label of $v$, left and right child of $v$ respectively. }
            \EndCase
            \Case{$p-vertex$}
                \State ProcessPvertex$(v)$
                \Comment{$(x,y)_p$ is the label of $v$ and the label of left and right child of $v$  is $(x,y)$. }
            \EndCase
            \Case{$g-vertex$}
                \State ProcessGvertex$(x,y,z)$
                \Comment{$(x,y)_s$, $(x,y)$ and $(y,z)$ are label of $v$, left and right child of $v$ respectively.}
            \EndCase
        \EndSwitch

\State $D\gets \emptyset$  \Comment{$(x,y)$ is the label of root $T$}
          \ForAll{$(i,j,k),(i',j',k')\in \{(1,0,0),(1,0,1),(2,0,0),(2,0,1)\}$}
                \State Add $S(x_{i,j,k},y_{i',j',k'})$ to $D$
          \EndFor
                \State $\gamma_{[1,2]}(G)\gets Minsize(D)$
\end{algorithmic}
\end{algorithm}
\begin{theorem}\label{Complexity}
	For a given GSP graph $G=(V,E)$, algorithm total $[1,2]$-MINSET find a minimum $[1,2]$-set in time $O(\vert V \vert)$.
	\end{theorem}
	\section*{Concluding remarks}
	
In this paper, we initiate the study of the $[1,2]$-set and total  $[1,2]$-set problems for generalized series-paralle graphs. We have also provided exact polynomial time algorithms for these classes of graphs.

\section*{Acknowledgment}
	The authors are grateful to A. Shakiba and A. K. Goharshady for their constructive comments
	and suggestions on improving of our paper.	
%%
%%\newpage
%%\bibliographystyle{acm}
%%\bibliography{mybib}

\end{document}